\documentclass[lettersize,journal]{IEEEtran}
\IEEEoverridecommandlockouts{}
\usepackage{amsfonts}
\usepackage{amssymb}
\usepackage{stfloats}
\usepackage{cite}
\usepackage{graphicx}
\usepackage{psfrag}
\usepackage{amsmath}
\usepackage{array}
\usepackage{epstopdf}
\usepackage{authblk}
\usepackage{graphicx} 
\usepackage{amsthm} 
\usepackage{lipsum}
\usepackage{verbatim} 
\usepackage{color}
\usepackage{caption}
\usepackage{subcaption}
\usepackage{amsmath,amsfonts,amssymb,mathrsfs}
\usepackage{bm}
\usepackage{dsfont}
\usepackage{upgreek}
\newtheorem{Remark}{Remark}

\usepackage[ruled,vlined]{algorithm2e}
\usepackage{stmaryrd}
\usepackage[normalem]{ulem} %

\newcommand{\add}[1]{{\color{blue}{#1}}}

\newtheorem{Definition}{Definition}
\newtheorem{Lemma}{Lemma}

\DeclareMathOperator*{\argmax}{arg\,max}
\DeclareMathOperator*{\argmin}{arg\,min}

\newcommand*{\vehSet}{\mathcal{N}}
\newcommand*{\waypSet}{\mathcal{W}}
\newcommand*{\vehNum}{N}

\newcommand*{\bsSet}{\mathcal{M}}
\newcommand*{\bsNum}{M}
\newcommand*{\txTime}{\tau}

\newcommand*{\periodLen}{{\mathsf{T}_p}}
\newcommand*{\slotLen}{{\mathsf{T}_s}}
\newcommand*{\state}{\mathsf{S}}
\newcommand*{\action}{\mathsf{A}}
\newcommand*{\bandwidth}{\mathrm{B}}
\newcommand*{\spIdx}{\iota}
\newcommand*{\spLen}{\Gamma}
\newcommand*{\spNum}{T^s}

\newcommand*{\nthVehicleFinishTime}{ t_f }

\usepackage[user]{zref}

\newcounter{pCounter}
\newcommand{\plabel}[1]{%
    \refstepcounter{pCounter}%
    \zlabel{#1}%
    \textbf{P\thepCounter}%
}

\newcommand{\pref}[1]{\textbf{P\zref{#1}}}

\newcounter{cCounter}
\newcommand{\clabel}[1]{%
    \refstepcounter{cCounter}%
    \zlabel{#1}%
    \textit{C\thecCounter}%
}

\newcommand{\cref}[1]{\textit{C\zref{#1}}}

\usepackage{etoolbox}
\newtoggle{blueeqn}

\makeatletter
\let\orig@tagform@\tagform@
\newcommand{\add@tagform}[1]{\textcolor{black}{\orig@tagform@{#1}}}
\renewcommand{\tagform@}[1]{%
  \iftoggle{blueeqn}{\add@tagform{#1}}{\orig@tagform@{#1}}%
}
\makeatother

\newenvironment{blueeqns}
{\toggletrue{blueeqn}%
  \ignorespaces%
}
{\togglefalse{blueeqn}%
  \ignorespacesafterend%
}

\DeclareCaptionLabelFormat{blue}{\color{blue}#1 #2} %
\DeclareCaptionFont{bluefont}{\color{blue}}

\DeclareCaptionLabelFormat{bluesub}{\color{blue}#1~\textup{(#2)}} %

\usepackage{environ} %

\def\BibTeX{{\rm B\kern-.05em{\sc i\kern-.025em b}\kern-.08em
    T\kern-.1667em\lower.7ex\hbox{E}\kern-.125emX}}

\begin{document}
\title{A Dynamic Programming Framework for Vehicular Task Offloading with Successive Action Improvement
}

\author{Qianren Li, Yuncong Hong, Bojie Lv, and Rui Wang\\
The Department of Electronic and Electrical Engineering \\
The Southern University of Science and Technology
\thanks{
A preliminary version of this work has been accepted by IEEE Wireless Communications and Networking Conference (IEEE WCNC) 2025 \cite{li2025dynamic}. In this journal version, a new two-time-scale framework is introduced. It is shown by simulation that the new framework achieves a better performance than the one in the conference version.

	Qianren Li, Bojie Lv, Yuncong Hong, and Rui Wang are with the Department of Electronic and Electrical Engineering (EEE), Southern University of Science and Technology (SUSTech), China. 	Email: liqr2022@mail.sustech.edu.cn, hongyc@mail.sustech.edu.cn, lyubj@mail.sustech.edu.cn, wang.r@sustech.edu.cn.
}}
\maketitle
\begin{abstract}
    In this paper, task offloading from vehicles with random velocities is optimized via a novel dynamic programming framework. Particularly, in a vehicular network with multiple vehicles and base stations (BSs), computing tasks of vehicles are offloaded via BSs to an edge server. Due to the random velocities, the exact locations of vehicles versus time, namely trajectories, cannot be determined in advance. Hence, instead of deterministic optimization, the cell association, uplink time, and throughput allocation of multiple vehicles during a period of task offloading are formulated as a finite-horizon Markov decision process. In order to derive a low-complexity solution algorithm, a two-time-scale framework is proposed. The scheduling period is divided into super slots, each super slot is further divided into a number of time slots. At the beginning of each super slot, we first obtain a reference scheduling scheme of cell association, uplink time and throughput allocation via deterministic optimization, yielding an approximation of the optimal value function. Within the super slot, the actual scheduling action of each time slot is determined by making improvement to the approximate value function according to the system state. Due to the successive improvement framework, a non-trivial average cost upper bound could be derived. In the simulation, the random trajectories of vehicles are generated from a high-fidelity traffic simulator. It is shown that the performance gain of the proposed scheduling framework over the baselines is significant.
\end{abstract}

\section{Introduction}
Edge or cloud computing for vehicular networks has become necessary in many emerging applications \cite{wangArchitecturalDesignAlternatives2020}. For example, the online ride-hailing services, such as Uber and Didi Chuxing, have gained widespread popularity. 
In order to address the safety risks \cite{UbersUSSafety}, the data of safety monitoring in the vehicle should be offloaded to remote servers. However, limited wireless radio resource for task offloading may hinder real-time response of edge or cloud computing. Moreover, the high mobility of vehicles may bring uncertainty in the uplink transmission scheduling, and hence, degrade the scheduling efficiency.

In fact, significant research efforts have been devoted to efficient scheduling of task offloading in vehicular networks \cite{jiaLearningBasedQueuingDelayAware2021,zhangJointOffloadingDecision2021, choEnergyEfficientCooperativeOffloading2022, bozorgchenaniComputationOffloadingHeterogeneous2022}.
For instance, the work in \cite{jiaLearningBasedQueuingDelayAware2021} modeled vehicle positions versus time via trajectory tables \cite{wangMaximizingSpatialTemporal2018}, which were  leveraged by the scheduler to optimize the offloading throughput. 
In contrast, the studies in \cite{zhangJointOffloadingDecision2021, choEnergyEfficientCooperativeOffloading2022, bozorgchenaniComputationOffloadingHeterogeneous2022} assumed constant vehicle velocities in the optimization of edge computing performance. Additionally, the work in \cite{shiPriorityAwareTaskOffloading2020} adopted a free-flow traffic model \cite{fricker2004fundamentals}, where vehicle velocities were deterministically related to traffic density.
In all these approaches, vehicle trajectories were assumed to be deterministic and fully known to the scheduler. 
Consequently, task offloading strategies were formulated as deterministic optimization problems. 
However, in real-world scenarios, even when vehicles follow predetermined routes, their velocities can exhibit randomness due to dynamic traffic conditions.

Incorporating the randomness of vehicle trajectories, task offloading scheduling should depend on the dynamic system state, yielding  a stochastic optimization problem. Hence, the methods of Lyapunov optimization \cite{yangDynamicTrajectoryOffloading2021} and Markov decision process (MDP) \cite{linPopularityAwareOnlineTask2022,zhangMDPBasedTaskOffloading2020, liDelayTolerantDataTraffic2018, malekiHandoverEnabledDynamicComputation2023, liuDeepReinforcementLearning2019} could be applied. For example in \cite{yangDynamicTrajectoryOffloading2021}, the average energy consumption was minimized subject to the queue stability constraint via the Lyapunov optimization, where the speeds of vehicles were modeled by the Gauss-Markov process. However, this method could not directly optimize or limit the performance of offloading delay. On the other hand, although the MDP formulation could address the above issues, the computation complexity of optimal solution is usually prohibitive \cite{zhangMDPBasedTaskOffloading2020, liDelayTolerantDataTraffic2018}. Hence, a number of approximate MDP methods were investigated. For example, the multi-armed bandits (MAB) method was proposed in \cite{linPopularityAwareOnlineTask2022}  to minimize the average offloading energy while ensuring the constraint on the task completion time, where the Manhattan model \cite{harriMobilityModelsVehicular2009} was adopted to model the random trajectories of vehicles. The MAB approach relies on the statistics learning to track the average future performance with vehicles' random trajectories. The efficiency might be low in the context that the trajectory statistics have already known to the scheduler. A model-assisted Q-learning algorithm was proposed in \cite{liuDeepReinforcementLearning2019} to optimize a joint communication and computation utility of vehicular task offloading, where the scheduler selected offloading server from a dynamic candidates set and allocated uplink spectrum for task offloading. A Double Deep Q-Network was proposed in \cite{malekiHandoverEnabledDynamicComputation2023} to minimize both average task completion time and average energy consumption, where the vehicles' mobility follows the Manhattan model and the dynamic server selection was considered. However, in these works, the complexity of training neural networks is still high, especially when the neural networks should be generally applied on a large number of traffic scenarios. Moreover, their performance could hardly be predicted from the perspective of algorithm robustness.

In this paper, we would like to shed some light on the low-complexity stochastic optimization for task offloading in a vehicular network with random trajectories. Particularly, we consider the scheduling of task offloading from multiple vehicles to an edge server via multiple base stations (BSs). Due to the randomness in the vehicles' velocities, we formulate the BS association, uplink time and throughput allocation for task offloading as a finite-horizon MDP. The contributions of this paper are summarized below.
\begin{itemize}
    \item A novel low-complexity solution, consisting of large-time-scale (super-slot-scale) deterministic pre-allocation and small-time-scale (slot-scale) online adjustment, is proposed. The pre-allocation is conducted periodically according to deterministic representative trajectories every super slot, which consists of a number of time slots. By approximating the optimal value function of the Bellman's equations \cite{bertsekas1996dynamic} analytically via the pre-allocation, the conventional value iteration can be eliminated, and the small-time-scale online scheduling can be optimized efficiently, instead of the exhaustive search.
    \item Due to the successive action improvement structure of the proposed solution, a non-trivial upper bound on the average scheduling cost could be derived.
    \item A high-fidelity traffic simulator is adopted to generate the trajectories of vehicles, and significant performance gain of the proposed solution is demonstrated in the simulation.
\end{itemize}

The remainder of this paper is organized as follows. In Section~\ref{sec: System Model}, the models of vehicles' mobility, uplink transmission, and task queuing are elaborated. 
In Section~\ref{sec: problem formulation}, the joint scheduling of uplink power, time, and BS association is formulated as a finite-horizon MDP.
Then, a novel two-time-scale solution framework based on successive action improvement is proposed in Section~\ref{sec: Proposed Framework}, where the optimization algorithm for super-slot-scale per-allocation is elaborated in Section~\ref{section: solution of offline pre-allocation} and the optimization algorithm for online scheduling is elaborated in Section~\ref{section: solution of dynamic improvement}, respectively.
The complexity and performance of the proposed solution is investigated in Section~\ref{section: complexity and performance analysis}, and the simulation results are demonstrated in Section~\ref{section: Performance Evaluation}. Finally, the conclusions are drawn in Section~\ref{sec: Conclusion}.

\textit{Notations}: The bold lowercase letters (\textit{e.g.}, $\boldsymbol{l}$) represent vectors, and the bold uppercase letters (\textit{e.g.}, $\boldsymbol{L}$) represent matrices. The element in the $i$-th row and $j$-th column of matrix $\boldsymbol{L}$ is denoted as $[\boldsymbol{L}]_{i,j}$; the $n$-th row of matrix $\boldsymbol{L}$ is denoted as $[\boldsymbol{L}]_n$.
Script symbols (\textit{e.g.}, $\mathcal{N}$) denote sets, and $\vert \mathcal{N} \vert$ indicates the cardinality of a set. Let $(\cdot)^+ = \max\{0, \cdot\}$, and $\llbracket S, T \rrbracket$ represent the set of integers from $S$ to $T$. 
The expectation operator is denoted by $\mathds{E}[\cdot]$, the probability of an event by $\mathds{P}[\cdot]$, and the indicator function by $\mathds{1}(\cdot)$. 
The set of real numbers is denoted by $\mathbb{R}$. 
We also use Landau notation, where $O(n)$ represents a function whose absolute value is bounded below by $cn$ for some constant $c > 0$.
Finally, a summary of notations are provided in Table~\ref{tab:notation}.

\renewcommand{\arraystretch}{1.2}
\begin{table}
	\centering
	\caption{Table of notations}	\label{tab:notation}
	\begin{tabular}{|c|c|}
		\hline
		\textbf{Parameter} & \textbf{Symbol}\\
		\hline
		Number/ Set of vehicles     & $N$, $\mathcal{N}$ \\ \hline
		Number/ Set of BSs          & $M$, $\mathcal{M}$ \\ \hline
	    $n$-th vehicle's Route     & $\mathcal{W}_n$ 	\\ \hline
	    Location of vehicle         & $\boldsymbol{l}_{n,t}$ \\ \hline 
	    Scheduling Period           & $\periodLen$ \\ \hline 
	    Slot duration               & $\slotLen$ \\ \hline
	    Super slot/ slot index      & $\iota$, $t$ \\ \hline
	    Association matrix          & $\boldsymbol{E}_t$ \\ \hline
	    System state                & $\state_t$ \\ \hline
	    Action                      & $\action_t$ \\ \hline
	    $t$-th Policy/ Overall Policy& $\Omega_t$, $\Omega$ \\ \hline
	    Reference Scheduling        & $\mathcal{R}_\iota$ \\ \hline
	    System Cost                 & $\mathsf{C}_t( \state_t, \Omega)$ \\ \hline
	\end{tabular}
\end{table}

\section{System Model}
\label{sec: System Model}

The task offloading of $\vehNum$ vehicles in a scheduling period of $\periodLen$ seconds via a cellular network of $\bsNum$ BSs is considered in this paper. The sets of vehicles and BSs are denoted as $\vehSet = \{ 1, 2, \ldots, \vehNum\}$ and $\bsSet = \{1,2,\ldots, \bsNum\}$, respectively. The vehicles traverse predetermined routes with random velocities. One computing task is generated at each vehicle during the scheduling period and uploaded to an edge server via the BSs. The edge server is connected to all the $M$ BSs. Since vehicles offload tasks in motion, handovers between BSs may occur during transmission. The BSs operate on different frequency bands, each with a bandwidth of $\bandwidth$ Hz. The vehicles' mobility model, uplink transmission model, and the tasks' queuing model are elaborated below.

\subsection{Vehicle Mobility Model}
The predetermined routes of all vehicles are known to a centralized offloading scheduler as in the popular ride-hailing services (e.g., Didi and Uber). To facilitate the offloading scheduling, each route is quantized into a sequence of waypoints. For example, the route of the $n$-th vehicle ($\forall n$) is defined as the ordered set of waypoints
\begin{blueeqns}
\begin{equation}
  \waypSet_n \triangleq \{\boldsymbol{l}^0_n, \boldsymbol{l}^1_n, \ldots, \boldsymbol{l}^{|\waypSet_n|-1}_n\},
\end{equation}
where $\boldsymbol{l}^0_n \in \mathbb{R}^{2 \times 1}$ is the initial position of the route of $n$-th vehicle at the beginning of the scheduling period. Each subsequent element, say $\boldsymbol{l}^i_n \in \mathbb{R}^{2 \times 1}$, represents the position of the $i$-th waypoint along the $n$-th vehicle's route.
\end{blueeqns}

The scheduling period of $\periodLen$ seconds is divided into $T$ time slots, and the slot duration is $\slotLen=\periodLen/T$. The position of the $n$-th vehicle versus time slot in the considered scheduling period, referred to as its trajectory, can be modeled as a stochastic process on $\waypSet_n$. Within the time scale of the considered scheduling period, the trajectory statistics are assumed to be stationary. As a remark, note that the trajectory statistics may vary in a larger time scale than the scheduling period. For example, a weekday morning should have a different traffic condition from a holiday morning. 

Let $\boldsymbol{l}_{n,t} \in \waypSet_n$ represent the location of $n$-th vehicle in the $t$-th slot. The trajectory of the $n$-th vehicle, denoted as 
\begin{blueeqns}
\begin{equation}
  (\boldsymbol{l}_{n,1}, \boldsymbol{l}_{n,2}, \ldots, \boldsymbol{l}_{n,T}),
\end{equation}
is modeled as a stationary Markov chain, where the transition probability matrix $\mathbf{P}_n \in \mathbb{R}^{\vert \waypSet_n \vert \times \vert \waypSet_n \vert}$ is defined as
\end{blueeqns}
\begin{equation}
 [\mathbf{P}_n]_{i,j} \triangleq \mathds{P}\left[\boldsymbol{l}_{n,t+1} = \boldsymbol{l}^{i}_n | \boldsymbol{l}_{n, t} = \boldsymbol{l}^j_n\right], \forall i, j,t,
\end{equation}
and $[\mathbf{P}_n]_{i,j}$ is the $(i,j)$-th element of matrix $\mathbf{P}_n$.
Moreover, the aggregation matrix $\boldsymbol{L}_t \in \mathbb{R}^{2 \times \vehNum}$ of $\vehNum$ vehicles' locations in $t$-th slot is given by
\begin{equation}
 \boldsymbol{L}_t \triangleq [\boldsymbol{l}_{1,t}, \boldsymbol{l}_{2,t}, \ldots, \boldsymbol{l}_{\vehNum,t}].
\end{equation}
It is assumed that the transition probability matrix $\mathbf{P}_n $ is known to the centralized offloading scheduler. In fact, a high-fidelity simulator has been proposed in \cite{zhang2021distributed} to obtain the traces of vehicles' trajectories for arbitrary road maps and traffic scenarios, such that the above transition matrix could be estimated based on the traces \cite[Sec. IV]{yunCongFedCars2024}.

\subsection{Uplink Model}
\label{section: Channel Model}

\begin{figure}[t]
  \centering
  \includegraphics[width=\linewidth]{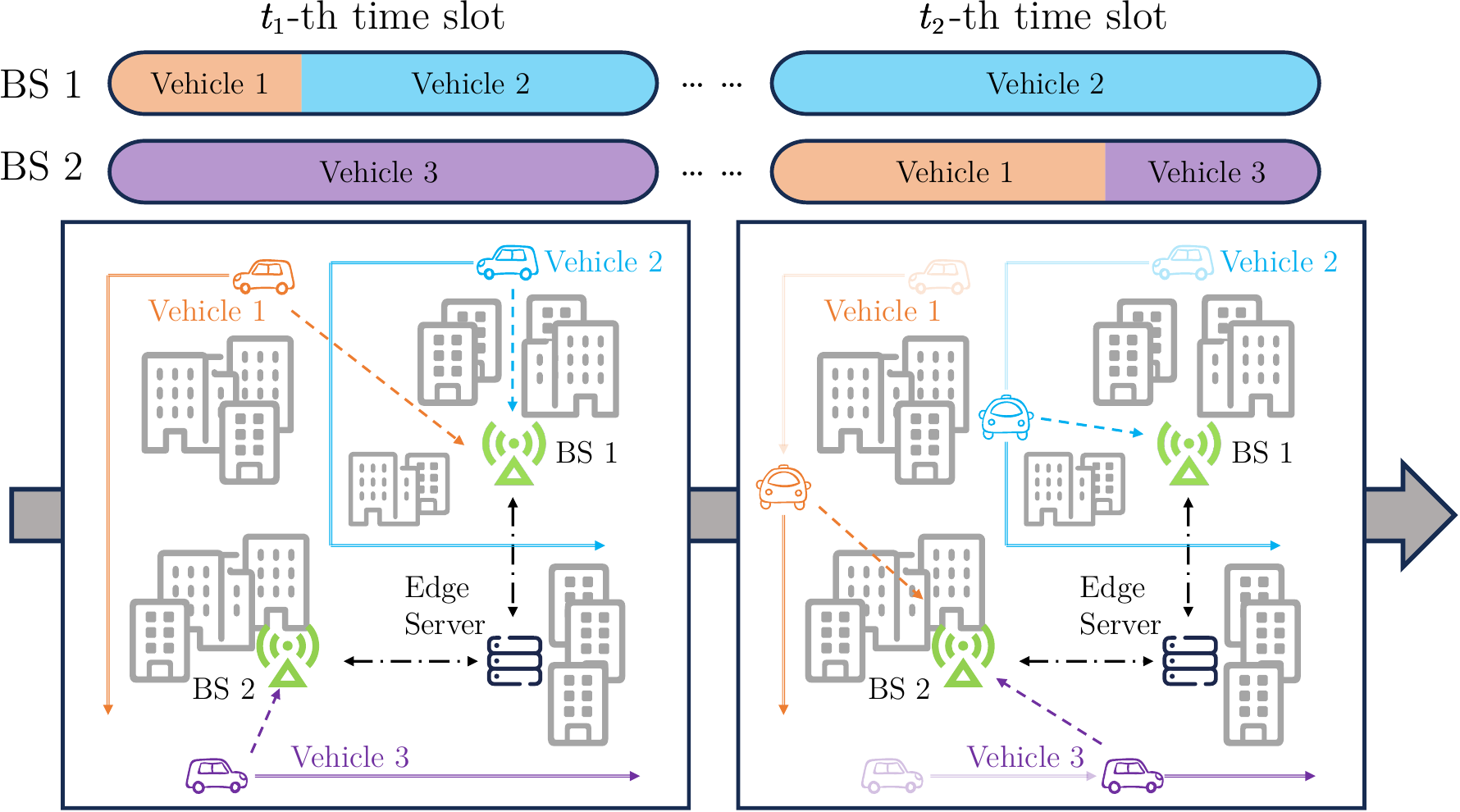}
  \caption{Conceptual illustration of vehicular task offloading. Three vehicles offload computation tasks to two BSs. The solid arrows represent the vehicles' routes, and the dashed arrows represent their BS associations. In the $t_1$-th time slot, the $1$-st and $2$-nd vehicles offload via the $1$-st BS; the $3$-rd vehicle offloads via the $2$-nd BS. In the $t_2$-th time slot, due to the vehicles' mobility, their BS association changes: the $1$-st and $3$-rd vehicles offload via the $2$-nd BS; the $2$-nd vehicle offloads via the $1$-st BS.}
  \label{fig: conceptual figure}
\end{figure}

In each time slot, each vehicle associates with at most one BS for the uplink transmission of its computing task. Let $e_{n,m,t} \in \mathscr{E} \triangleq \{0, 1\}$, $\forall n,m,t$, be the indicator of BS association, where $e_{n,m,t} = 1$ indicates the association of the $n$-th vehicle with the $m$-th BS in the $t$-th time slot and $e_{n,m,t} = 0$ otherwise.
Define the association matrix of the $t$-th time slot $\boldsymbol{E}_t \in \mathscr{E}^{\vehNum \times \bsNum}$ as
$
 \left[\boldsymbol{E}_t\right]_{n,m} \triangleq e_{n,m,t}.
$
The constraint of unique BS association can be written as
\begin{equation}
 \clabel{BS Association}: 
 \boldsymbol{E}_t \mathbf{1} = \mathbf{1}, \forall t,
\end{equation}
where $\mathbf{1}$ is the column vector of 1s.

The time-division multiple access (TDMA) mechanism is adopted in each cell. Let $\txTime_{n,t}$ be the uplink time ratio of the associated cell allocated to the $n$-th vehicle in the $t$-th time slot, and $ \boldsymbol{\txTime}_t \triangleq [\txTime_{1,t}, \txTime_{2,t}, \ldots, \txTime_{\vehNum,t}]^T \in \mathbb{R}^{\vehNum\times 1} $ be the aggregation vector of time ratios. The uplink transmission time of the $n$-th vehicle in the $t$-th time slot is $\txTime_{n,t} \slotLen$, and the transmission time constraints of all $M$ cells can be written as
\begin{equation}
 \clabel{Uplink Time Constraint 1}: \mathbf{0} \preceq \boldsymbol{\txTime}_t, \forall t, \quad
 \clabel{Uplink Time Constraint 2}:  \boldsymbol{\txTime}_t^T \boldsymbol{E}_t \preceq \mathbf{1}, \forall t,
\end{equation}
where $\preceq$ denotes the element-wise smaller or equal relationship, and $\mathbf{0} $ is the column vector of 0s.

Let $\{\boldsymbol{l}_1,...,\boldsymbol{l}_\bsNum\}$ be the locations of $\bsNum$ BSs, respectively. In the $t$-th time slot, given the locations of $n$-th vehicle, the path loss of its uplink channel to the $m$-th BS is modeled as 
\begin{equation}
 g_{n,m,t} = G \frac{1}{\| \boldsymbol{l}_{n,t} - \boldsymbol{l}_m \|_2^\gamma},
\end{equation}
where $G$ is a constant, and $\gamma$ denotes the path loss exponent. It is assumed that the time slot duration $\slotLen$ is sufficiently large, such that the ergodic channel capacity can be achieved. 
Although the uplink signals of the considered vehicles will not interfere each other due to orthogonal spectrum allocation, there is still co-channel interference from mobile users in the nearby cells with the same uplink spectrum. It is assumed that the aggregation of co-channel interference at each BS is stationary. Denote the aggregation of co-channel interference at the $m$-th BS as $I_m$.
Hence, the throughput of the $n$-th vehicle in the $t$-th time slot (if it is associated with the $m$-th BS) can be written as 
\begin{equation}
    \label{equation: Throughput}
 r_{n,t} = \txTime_{n,t} \slotLen \bandwidth \mathds{E}_{h,I_m} \left[\log_2 \left(1 + \frac{P_{n, t} g_{n,m,t} \vert h \vert^2}{\sigma^2 + I_m}\right) \right],
\end{equation}
where $\sigma^2$ is the average noise power, the expectation is taken over the fast fading coefficient $h$ and the co-channel interference $I_m$, and $P_{n,t}$ is the uplink power with the peak power constraint $P_{\max}$. Thus,
\begin{equation}
 \clabel{Peak Power Constraint}: 0 \leq P_{n,t} \leq P_{\max}, \forall n, t.
\end{equation} 

\subsection{Task Queuing Model}

In practice, the task arrivals of the vehicles are not synchronous. It is assumed that the computing task of the $n$-th vehicle arrives at the beginning of the $T_{n,a}$-th time slot, with a size of $D_{n,a}$ information bits. Hence, the vehicle initiates uplink transmission in the $T_{n,a}$-th time slot. The arrival tasks are buffered at the transmitters. Let $d_{n,t}$ be the number of information bits buffered at the $n$-th vehicle at the beginning of the $t$-th time slot. The uplink queue dynamics can be written as
\begin{equation}
 d_{n, t+1} = \begin{cases}
      0,                                   & t < T_{n,a}, \\
      \left(D_{n,a} - r_{n,t}\right)^+,      & t = T_{n,a}, \\
      \left(d_{n, t} - r_{n,t}\right)^+,   & \text{otherwise},
    \end{cases}
\end{equation}
where $\left( \cdot \right)^+=\max \{0, \cdot \}$. 

As a remark, notice that we focus on the offloading scheduling in a scheduling period of $\periodLen$ seconds. In practice, the time horizon can be divided into a sequence of scheduling periods, and our design can be applied on each scheduling period respectively. Different scheduling periods might be with different sets of vehicles in task offloading. Due to the late arrival, some tasks might not be completely offloaded in one scheduling period. They can be treated as new tasks in the next scheduling period. 

\section{Problem Formulation}
\label{sec: problem formulation}
The joint scheduling of BS association, time allocation and uplink power adaptation (or equivalently throughput adaptation) of all vehicles in all $T$ time slots is considered in this paper. Due to the randomness in the vehicles' mobility, the uplink channel cannot be predicted precisely in advance. This uncertainty necessitates a stochastic optimization approach. We therefore formulate the joint scheduling problem as a finite-horizon MDP. The system state, action, and policy are defined in the following.
\begin{Definition}[System State]
  \label{def:state}
 The system state in the $t$-th time slot ($\forall t \in \llbracket 1, T \rrbracket$) is a tuple of buffered information bits and the locations of all the vehicles, i.e.,
  \begin{equation}
 \state_t = (\boldsymbol{d}_t, \boldsymbol{L}_t),
  \end{equation}
 where vector $\boldsymbol{d}_t \triangleq [d_{1,t}, d_{2,t}, \ldots, d_{N,t}]$ denotes the buffered information bits for all vehicles.
\end{Definition}

\begin{Definition}[Policy]
  \label{def:policy}
 The action in the $t$-th time slot ($\forall t \in \llbracket 1, T \rrbracket$), denoted as $\action_t $, consists of the BS association, uplink time ratio, and uplink throughput\footnote{We define the uplink throughput, instead of uplink power, as the scheduling action for convenience of elaboration. In fact, they are equivalent.} of each vehicle, thus,
\begin{blueeqns}
\begin{equation}
  \action_t = (\boldsymbol{E}_t, \boldsymbol{\txTime}_t, \boldsymbol{r}_t),
\end{equation}
\end{blueeqns}
 where $\boldsymbol{r}_t\triangleq [r_{1,t}, r_{2,t},..., r_{N,t}]^T$ is the vector of throughput allocation for all the vehicles. 
 The policy $\Omega_t$ of the $t$-th time slot ($\forall t \in \llbracket 1, T \rrbracket$) is a mapping from the system state to the action, $\Omega_t: \state_t \mapsto \action_t$. The overall policy is defined as the aggregation of the policies of all the time slots, i.e.,
 \begin{blueeqns}
  \begin{equation}
    \Omega = ( \Omega_1, \Omega_2, \ldots, \Omega_T ).
  \end{equation}
  \end{blueeqns}
\end{Definition}

The joint scheduling design of this paper is to minimize the offloading latency of all tasks, while saving the uplink energy consumption of all vehicles. Hence, we consider a weighted summation of non-empty buffer penalty and uplink energy consumption as the system cost per time slot. Moreover, it is possible that the task offloading of one vehicle may not be accomplished at the end of the scheduling period, especially with a late task arrival or poor uplink channel condition. Hence, an additional penalty for unaccomplished offloading is considered in the $T$-th (last) time slot. Therefore, the cost function for $n$-th vehicle in the $t$-th time slot is written as
\begin{equation}
  \label{equation: per-slot cost function}
 \mathsf{c}_{n,t} (\state_{t}, \action_{t}) \triangleq \begin{cases}
 \mathds{1}( d_{n,t} > 0 ) + \omega_1 P_{n,t} \txTime_{n,t} ,         & t < T , \\
  \begin{aligned}
    &\mathds{1}( d_{n,t} > 0 ) + \omega_1 P_{n,t} \txTime_{n,t} \\
    &~+ \omega_2 d_{n,T+1},
  \end{aligned} & t = T,
  \end{cases}
\end{equation}
where the indicator function $\mathds{1}(\mathcal{C})$ is $1$ when the event $\mathcal{C}$ happens and $0$ otherwise, $\omega_1$ and $\omega_2$ are two weights, and $\slotLen$ is absorbed into $\omega_1$ for notation convenience.
A larger value of $\omega_1$ prioritizes the minimization of energy consumption, while a larger $\omega_2$ increases the penalty for residual information bits at the last time slot. In fact, these two parameters can also be considered as the Lagrange multipliers of a MDP with the constraints on the average energy consumption and the average number of residual information bits. By adjusting the values of $\omega_1$ and $\omega_2$, the Pareto frontier of average transmission time, average energy consumption and average number of residual information bits can be achieved.

The system cost of the $t$-th time slot is then given by 
\begin{blueeqns}
\begin{equation}
  \mathsf{c}_{t} (\state_{t}, \action_{t}) = \sum_{n} \mathsf{c}_{n,t} (\state_{t}, \action_{t}).
\end{equation}
Hence, the average total cost of the system from the $t$-th time slot can be written as
\end{blueeqns}
\begin{equation}
  \label{equation:objective}
 \mathsf{C}_t( \state_t, \Omega ) \triangleq \mathds{E}_{\state_{t+1}, \ldots, \state_T} \left[ \sum_{k = t}^{T} \mathsf{c}_{k} (\state_{k}, \action_{k}) \middle| \state_t , \Omega\right], 
\end{equation}
and the average total cost in one scheduling period can be written as $\mathsf{C}_1( \state_1, \Omega )$.
As a result, the joint scheduling problem can be formulated as a finite-horizon MDP as follows:
\begin{equation*}
  \begin{aligned}
    \plabel{Value Function}: \min_{\Omega }  \quad  & \mathsf{C}_1( \state_1, \Omega ), \\
    \text{s.t.} \quad & \cref{BS Association}, \cref{Uplink Time Constraint 1}, \cref{Uplink Time Constraint 2}, \cref{Peak Power Constraint}.
  \end{aligned}
\end{equation*}

Finding the optimal solution of \pref{Value Function} is equal to solving the following Bellman's equation~\cite{bertsekas1996dynamic} in each time slot:
\begin{equation}
  \label{equation: Bellman}
 V_{t}(\state_t) = \min_{\action_t} \! \left\{\! \mathsf{c}_t(\state_t, \action_t) \! + \! \sum_{\state_{t+1}} \mathds{P}\left[ \state_{t+1} | \state_{t}, \action_t \right] V_{t+1}(\state_{t+1}) \!\right\},
\end{equation}
where the optimal value function $V_t(\state_t)$ is the minimum average cost since the $t$-th time slot with the system state $\state_t$. Thus,
\begin{equation}
  \label{definition: optimal value function}
 V_t(\state_t) = \min_{\Omega_t,\ldots,\Omega_T }\mathds{E}_{\state_{t+1}, \ldots, \state_T} \left[ \sum_{k=t}^{T} \mathsf{c}_k (\state_k, \action_k) \middle| \state_t \right], \forall t.
\end{equation}
Particularly, the optimal action of the $t$-th time slot ($\forall t \in \llbracket 1, T \rrbracket$) can be obtained by solving the right-hand side (RHS) of the above Bellman's equations as
\begin{equation}\label{equation: Bellman-RHS}
\Omega_t^{\ast}(\state_t)=\argmin_{\action_t}  \mathsf{c}_t(\state_t, \action_t) + \sum_{\state_{t+1}} \mathds{P}\left[ \state_{t+1} | \state_{t}, \action_t \right] V_{t+1}(\state_{t+1}).
\end{equation}

\section{Successive Improvement Solution Framework}
\label{sec: Proposed Framework}

\begin{figure}
    \centering
    \includegraphics[width=\columnwidth]{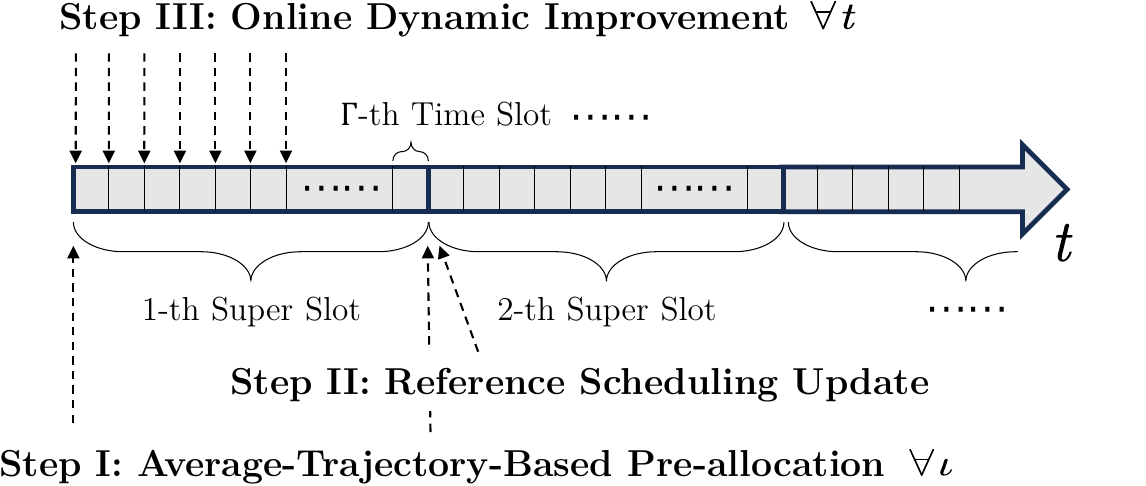}
    \caption{Illustration of the proposed dynamic improvement solution.}
    \label{figure: framework}
\end{figure}

The optimal solution of the finite-horizon MDP usually suffers from the {\it curse of dimensionality} \cite{powellApproximateDynamicProgramming2007}, where the computation complexity increases exponentially with the dimension of states (e.g., the number of vehicles). Particularly, the complexity of optimal solution lies in the evaluation of optimal value functions defined in \eqref{definition: optimal value function}, which are determined by the optimal scheduling policies.
To reduce the optimization complexity, an attractive approach is to approximate the optimal value function at the RHS of the Bellman's equation, and then the online scheduling based on the RHS of the Bellman's equation can be subsequently applied. Note that the approximation without physical meaning could hardly be analyzed, e.g., linear approximation \cite{rashid2020monotonic, sunehag2017value}. In this paper, we would address the approximation with physical meaning.

In our prior work \cite{li2025dynamic}, deterministic scheduling actions are used to evaluate the future system cost, serving as the approximations of the optimal value functions. Although such approximation is with physical meaning and leads to an achievable average total system cost, the deterministic scheduling actions may lead to low performance in a stochastic system. Hence, in this section, a novel successive improvement framework with two-time-scale optimization, as illustrated in \figurename~\ref{figure: framework}, is proposed to achieve a better solution of approximate MDP (AMDP).

In the proposed framework, every successive $\spLen$ time slots is organized as a super slot. Hence, there are $\spNum= T/\spLen$ super slots within the scheduling period. Then, deterministic scheduling actions for the remaining time slots are updated at the beginning of each super slot, which is referred to as the reference scheduling. The reference scheduling is used to approximate the optimal value function, such that the actual scheduling action of each time slot in the super slot can be optimized according to \eqref{equation: Bellman-RHS}. As a remark, note that the periodic update of reference scheduling in each super slot is to suppress the loss of the value function approximation. As a result, the proposed successive improvement solution framework consists of the following iterative steps.

{\bf Step I (Average-Trajectory-Based Pre-Allocation):} At the beginning of each super slot (say the $\spIdx$-th super slot), the BS association, time allocation, and throughput allocation of all the remaining time slots are optimized according to (deterministic) average trajectories of vehicles, yielding a deterministic optimization problem. Denote the solution as $\mathcal{R}_{\spIdx}^{pre}$, which is defined in Section~\ref{section: offline pre-allocation}. As a remark, note that the optimization solution $\mathcal{R}_{\spIdx}^{pre}$ is the candidate for value function approximation, rather than actual scheduling.

{\bf Step II (Reference Scheduling Update):} Let $\mathcal{R}_{\spIdx}$ be the reference scheduling of the $\spIdx$-th super slot. If $\spIdx=1$, the above optimization solution $\mathcal{R}_{\spIdx}^{pre}$ is used as the reference scheduling, denoted as $\mathcal{R}_{1}=\mathcal{R}_{1}^{pre}$. Otherwise, the average system cost of $\mathcal{R}_{\spIdx}^{pre}$ is compared with the revised reference scheduling of the $(\spIdx-1)$-th super slot, denoted as $\mathcal{R}_{\spIdx-1}^{post}$ (see Step III). The one with the less average system cost is treated as the reference scheduling of the $\spIdx$-th super slot $\mathcal{R}_{\spIdx}$.

{\bf Step III (Online Scheduling):} The above reference scheduling $\mathcal{R}_{\spIdx}$ is used to approximate the optimal value functions in the $\spIdx$-th super slot. Then, the actual scheduling action of each time slot in the $\spIdx$-th super slot can be obtained by solving \eqref{equation: Bellman-RHS}. At the end of the super slot, the reference scheduling $\mathcal{R}_{\spIdx}$ is revised to $\mathcal{R}_{\spIdx}^{post}$ for the next iteration. If $\spIdx=\spNum$, the iteration terminates; otherwise, $\spIdx=\spIdx+1$ and jumps back to the first step.

The comparison made in the second step is to ensure that the performance of reference scheduling in one super slot is not worse than the one in the previous super slot. This property will be exploited to prove an upper bound of the average system cost of the proposed solution. Moreover, due to the online scheduling, the reference scheduling determined at the beginning of one super slot may not be feasible at the end of the same super slot. The revision in the third step is to ensure the feasibility of the reference scheduling. The above three steps are further elaborated in the following three parts, respectively.

\subsection{Average-Trajectory-Based Pre-allocation}
\label{section: offline pre-allocation}

Without loss of generality, the average-trajectory-based pre-allocation in the $\spIdx$-th super slot ($\spIdx \in \llbracket 1, \spNum \rrbracket$) is elaborated in this part. 
The pre-allocation is conducted by assuming that all the vehicles move along their deterministic average trajectories. 
Then, the BS association, time, and throughput allocation for all the remaining time slots can be optimized via a deterministic optimization problem, which is subsequently considered as the reference scheduling in the following online scheduling. 

The average trajectory of the $n$-th vehicle since the $\spIdx$-th super slot can be expressed as 
\begin{blueeqns}
\begin{equation}
    \boldsymbol{L}^{(\spIdx)}_n = \left(\boldsymbol{l}_{n,(\iota-1)\spLen+1}^{(\spIdx)},...,\boldsymbol{l}_{n,T}^{(\spIdx)}\right),
\end{equation}
\end{blueeqns}
where
\begin{equation}
    \boldsymbol{l}_{n,t}^{(\spIdx)} \triangleq \mathds{E}\left[ \boldsymbol{l}_{n,t} | \boldsymbol{l}_{n, (\iota-1)\spLen + 1 }\right], \forall t=(\iota-1)\spLen + 1,...,T,
\end{equation}
represents the average location of the $n$-th vehicle at the $t$-th time slot. Moreover, the aggregation of average trajectories of all vehicles is denoted as $\mathcal{L}_\spIdx \triangleq \{ \boldsymbol{L}^{(\spIdx)}_n | n \in \mathcal{N}\}$.

Note that above average trajectories are deterministic. Let
\begin{blueeqns}
\begin{align}
	\mathcal{E}_{\iota } \triangleq &\left\{ \boldsymbol{E}_{(\iota-1)\spLen + 1  }, \ldots, \boldsymbol{E}_T \right\}, \\
	\mathcal{T}_{\iota } \triangleq &\left\{ \boldsymbol{\txTime}_{(\iota-1)\spLen + 1  }, \ldots, \boldsymbol{\txTime}_T \right\}, \\
	\Upsilon_{\iota } \triangleq &\left\{ \boldsymbol{r}_{(\iota-1)\spLen + 1 }, \ldots, \boldsymbol{r}_T \right\},
\end{align}
\end{blueeqns}
be the aggregation of BS association, time and throughput allocation for all the remaining time slots, assuming all vehicles are moving along the above average trajectories. In this case, the expectation of the random trajectories in \eqref{equation:objective} can be omitted, and the stochastic optimization of throughput allocation in \pref{Value Function} can be simplified to the following deterministic optimization problem:

\begin{equation*}
	\begin{aligned}
		\plabel{Reference Policy}:
		\min_{ 
			\substack{ 
				\mathcal{E}_{\iota }, 
                \mathcal{T}_{\iota }, \\ 
                \Upsilon_{\iota }
			}}  
			\quad  &\sum_{t=(\iota-1)\spLen + 1  }^{T} \sum_{n \in \mathcal{N}} \mathsf{c}_{n,t} \left(
                d_{n,t}, \boldsymbol{l}_{n,t}^{(\spIdx)}, [\boldsymbol{E}_t]_n, \txTime_{n,t}, r_{ n,t}
			\right), \\
		\text{s.t.} \quad &\cref{BS Association}, \cref{Uplink Time Constraint 1}, \cref{Uplink Time Constraint 2}, \cref{Peak Power Constraint}.
	\end{aligned}
\end{equation*}
where $[\boldsymbol{E}_t]_n$ denotes the $n$-th row vector of $\boldsymbol{E}_t$. The solution of \pref{Reference Policy} is denoted as $\mathcal{R}_{\spIdx}^{pre} \triangleq (\mathcal{E}_{\iota }^\ast, \mathcal{T}_{\iota }^\ast, \Upsilon_{\iota }^\ast)$. In order to keep the elaboration of solution framework clear, the solution algorithm for \pref{Reference Policy} will be elaborated separately in Section~\ref{section: solution of offline pre-allocation}. 

\begin{table*}[!t]
    \centering
    \caption{Definition of $f_{n, t+1}^1(d_{n,t+1}, \varUpsilon_{n,t+1})$ and $f_{n, t + 1}^2(d_{n,t+1}, \boldsymbol{l}_{n,t+1} , \varUpsilon_{n,t+1})$.}
    \label{table: case function}
    \begin{tabular}{c|c|l}
    \hline
     Condition  & $f_{n, t+1}^1(d_{n,t+1}, \varUpsilon_{n,t+1})$ & $f_{n, t + 1}^2(d_{n,t+1}, \boldsymbol{l}_{n,t+1} , \varUpsilon_{n,t+1})$  \\ \hline
     $d_{n,t+1} \le 0$ &  $0$      &$0$              \\ \hline
    $ 
    \begingroup
        \renewcommand*{\arraystretch}{1.5}
        \begin{matrix}
            \sum_{\kappa= 1}^{k} r_{n,t+k} < d_{n,t+1} \le \sum_{\kappa= 1}^{k+1} r_{n,t+k}, \\
            k \in \llbracket 1, T-t-1 \rrbracket
        \end{matrix}
    \endgroup
    $ 
    & 
    $k+1$ 
    &        
    $
    \begin{aligned}
        & \sum_{\kappa= 1}^{k }\txTime_{n,t+\kappa}^{(\spIdx)} \varPhi_{n, t+\kappa|t+1} 2^{\frac{r_{n,t+\kappa}}{\txTime_{n,t+\kappa}^{(\spIdx)} \slotLen \bandwidth}}\\
        & + \txTime_{n,t+k+1}^{(0)} \varPhi_{n, t+k+1|t+1} 2^{\frac{d_{n, t} - \sum_{\kappa= 1}^{k} r_{n,t+k}}{ \txTime_{n,t+k+1}^{(\spIdx)} \slotLen \bandwidth }}
    \end{aligned}
    $ \\ \hline
    $\sum_{\kappa= 1}^{T - t} { r_{n,t+\kappa} } < d_{n,t+1}$  & $T - t + \omega_2\left(d_{n,t+1} - \sum_{\kappa= 1}^{T - t} {r_{n,t + \kappa} }\right)$  & $ \sum_{\kappa= 1}^{T-t}\txTime_{n,t+\kappa}^{(\spIdx)} \varPhi_{n, t+\kappa|t+1} 2^{\frac{r_{n,t+k}}{\txTime_{n,t+\kappa}^{(\spIdx)}\slotLen \bandwidth}}$     \\ \hline
    \end{tabular}
\end{table*}

\subsection{Reference Scheduling Update}

As mentioned above, at the beginning of the $\spIdx$-th super slot, a scheduling policy $\mathcal{R}_{\spIdx}^{pre}$ based on average trajectories is derived. In the first super slot ($\spIdx = 1$), $\mathcal{R}_{1}^{pre}$ is directly used as the reference scheduling. In the following $\spIdx$-th super slot ($\spIdx>1$), $\mathcal{R}_{\spIdx}^{pre}$ is compared with the revised reference scheduling of the previous super slot, denoted as $\mathcal{R}_{\spIdx-1}^{post}$. The one with the lower average total cost is chosen as the reference scheduling of the $\spIdx$-th super slot, denoted as $\mathcal{R}_{\spIdx}$.

Specifically, let $\action_{t,\spIdx}^{pre}$ and $\action_{t,\spIdx-1}^{post}$
be the scheduling actions (including the BS association, time and throughput allocation) of the $t$-th time slot in $\mathcal{R}_{\spIdx}^{pre}$ and $\mathcal{R}_{\spIdx-1}^{post}$, respectively. The average total cost of $\mathcal{R}_{\spIdx}^{pre}$ and $\mathcal{R}_{\spIdx-1}^{post}$ is given by
\begin{equation}
    \mathsf{C}_t( \state_t, \mathcal{R}_{\spIdx}^{pre}) = \mathds{E}_{\state_{t+1}, \ldots, \state_T} \left[ \sum_{k = t}^{T} \mathsf{c}_{k} (\state_{k}, \action_{k,\spIdx}^{pre}) \middle| \state_t\right],
\end{equation}
and 
\begin{equation}
    \mathsf{C}_t( \state_t, \mathcal{R}_{\spIdx-1}^{post}) = \mathds{E}_{\state_{t+1}, \ldots, \state_T} \left[ \sum_{k = t}^{T} \mathsf{c}_{k} (\state_{k}, \action_{k,\spIdx-1}^{post}) \middle| \state_t\right].
\end{equation}
Hence, the reference scheduling of the $\spIdx$-th super slot can be determined as
\begin{equation}
    \label{equation: reference scheduling update}
    \begin{aligned}
        \mathcal{R}_{\spIdx} = \argmin_{\mathcal{R}\in \{\mathcal{R}_{\spIdx}^{pre},\mathcal{R}_{\spIdx-1}^{post}\}}  \mathsf{C}_t( \state_t, \mathcal{R}).
    \end{aligned}
\end{equation}

\subsection{Online Scheduling}

For convenience of elaboration, denote the reference scheduling of the $\spIdx$-th super slot by 
\begin{blueeqns}
\begin{equation}
    \mathcal{R}_{\spIdx}=\{ \action_{t,\spIdx}^{ref}|t=(\iota-1)\spLen + 1,...,T\},
\end{equation}
\end{blueeqns}
where the scheduling action of the $t$-th time slot is
\begin{equation}
    \label{equation: reference action}
    \action_{t,\spIdx}^{ref} = \left( \boldsymbol{E}_{t}^{(\spIdx)}, \boldsymbol{\txTime}_{t}^{(\spIdx)}, \boldsymbol{r}_{t}^{(\spIdx)} \right).
\end{equation}
Then, the online improvement via \eqref{equation: Bellman-RHS} is applied to each time slot of this super slot (say the $t$-th time slot, $t=(\iota-1)\spLen + 1,...,\iota\spLen$), where the optimal value function is approximated based on  $\mathcal{R}_{\spIdx}$.
Particularly, the approximation of optimal value function is given below.

\begin{Lemma}[Asymptotically Achievable Average Cost]\label{lemma: Asymptotically Achievable Average Cost} Denote the aggregation of the throughput allocation of the $n$-th vehicle since the $t$-th time slot as
\begin{equation}
        \varUpsilon_{n,t+1} \triangleq (r_{n,t+1}, r_{n,t+2}, ..., r_{n,T}).\label{symbol: throughput allocation}
\end{equation}
Given the number of residual information bits  $d_{n,t+1}$, the position $\boldsymbol{l}_{n,t+1}$, the throughput allocation $\varUpsilon_{n,t+1}$, and sufficiently large peak power constraint, following the BS association and time allocation of reference scheduling $\mathcal{R}_{\spIdx}$, the average cost of the $n$-th vehicle since the $(t+1)$-th time slot can be written as
    \begin{multline}
        \label{equation: achievable cost}
        \widetilde{V}_{n,t+1}\left(
            d_{n,t+1},\boldsymbol{l}_{n,t+1},
            \varUpsilon_{n,t+1}
            \right) = f_{n, t+1}^1(d_{n,t+1}, \varUpsilon_{n,t+1}) \\  + f_{n, t + 1}^2(d_{n,t+1}, \boldsymbol{l}_{n,t+1} , \varUpsilon_{n,t+1}), 
    \end{multline}
    where the functions $ f_{n, t+1}^1 $ and $f_{n, t+1}^2 $ are defined in Table \ref{table: case function}, and
    \begin{blueeqns}
    \begin{equation}
        \varPhi_{n, t+k|t+1} \triangleq 2^{-  \mathds{E}\left[\log_2 \frac{G\vert h \vert^2}{\sigma^2}\right]} \mathds{E} [\Vert \boldsymbol{l}_{n,t+k} - \boldsymbol{l}_m \Vert_2^{\gamma}| \boldsymbol{l}_{n,t+1}].
    \end{equation}
    \end{blueeqns}
\end{Lemma}
\begin{proof}
    Please refer to Appendix \ref{appendix: achievable cost}.
\end{proof}

Hence, the expected value function $\sum_{\state_{t+1}} \mathds{P}\left[ \state_{t+1} | \state_{t}, \action_t \right] V_{t+1}(\state_{t+1})$ in \eqref{equation: Bellman-RHS} is approximated as 
\begin{equation}
    \begin{aligned}
    \label{eq: Q-factor}
    &\sum_{\state_{t+1}} \mathds{P}\left[ \state_{t+1} | \state_{t}, \action_t \right] V_{t+1}(\state_{t+1}) \\
    \approx &~ \widetilde{Q}_{t}(\state_{t}, \action_t; \mathcal{R}_{\spIdx})  \\ 
    \triangleq &\min_{\{\varUpsilon_{n,t+1} | n \in \mathcal{N} \}} \sum_{\state_{t+1}} \mathds{P}\left[ \state_{t+1} | \state_{t}, \action_t \right] \\
    &\times \sum_{n \in \mathcal{N}} \widetilde{V}_{n,t+1}\left(
            d_{n,t+1},\boldsymbol{l}_{n,t+1},
            \varUpsilon_{n,t+1}\right).
    \end{aligned}
\end{equation}
As a result, the scheduling action of the $t$-th time slot can be approximately obtained via the following optimization problem.
\begin{equation*}
\begin{aligned}
\plabel{Online Optimization Problem}: \action_t^*=&\argmin_{ \action_t }\!\!\!\!&& 
\mathsf{c}_t(\state_t, \action_t) + \widetilde{Q}_{t}(\state_{t}, \action_t; \mathcal{R}_{\spIdx}) \\ 
= &\argmin_{\substack{\action_t,\\\{\varUpsilon_{n,t+1} | n \in \mathcal{N} \}}} \!\!\!\!&& \mathsf{c}_t(\state_t, \action_t) + \sum_{\state_{t+1}} \mathds{P}\left[ \state_{t+1} | \state_{t}, \action_t \right] \\
& &&\times \!\sum_{n \in \mathcal{N}} \widetilde{V}_{n,t+1}\left(
        d_{n,t+1},\boldsymbol{l}_{n,t+1},
        \varUpsilon_{n,t+1}\right).\\
&\quad\quad\ \text{s.t.}  && \cref{BS Association}, \cref{Uplink Time Constraint 1}, \cref{Uplink Time Constraint 2}, \cref{Peak Power Constraint}. 
\end{aligned}
\end{equation*}

The solution of the above online successive improvement will be elaborated in Section~\ref{section: solution of dynamic improvement}, and denote optimal online scheduling as 
\begin{blueeqns}
\begin{equation}
    \widetilde{\Omega} \triangleq \{ \action_t^*, t = 1, ..., T\}.
\end{equation}
\end{blueeqns}

At the end of the $\spIdx$-th super slot, the reference scheduling $\mathcal{R}_{\spIdx}$ is revised to $\mathcal{R}_{\spIdx}^{post}$ as follows for the next super slot:
\begin{equation}
    \mathcal{R}_{\spIdx}^{post} = \{(\boldsymbol{E}_{t}^{(\spIdx)}, \boldsymbol{\txTime}_{t}^{(\spIdx)}, \boldsymbol{r}_{t, \iota}^{post})|t=\iota\spLen + 1, ..., T\},
\end{equation}
where $\boldsymbol{E}_{t}^{(\spIdx)}$ and $\boldsymbol{\txTime}_{t}^{(\spIdx)}$ are defined in \eqref{equation: reference action}, and
\begin{eqnarray}
    &&\boldsymbol{r}_{t, \iota}^{post} \nonumber\\
    & \triangleq& \left\{\left(r_{n,t+1}^{post}, r_{n,t+2}^{post}, ..., r_{n,T}^{post}\right) \middle|\forall n \in \mathcal{N} \right\} \nonumber\\
    &=& \argmin_{\{\varUpsilon_{n,\iota\spLen + 1} | \forall n \in \mathcal{N} \}} \sum_{\state_{\iota\spLen + 1}} \mathds{P}\left[ \state_{\iota\spLen + 1} | \state_{\iota\spLen}, \action_{\iota\spLen}^\ast \right] \nonumber \\
    &&\times \sum_{n \in \mathcal{N}} \widetilde{V}_{n,\iota\spLen + 1}\left(
        d_{n,\iota\spLen + 1},\boldsymbol{l}_{n,\iota\spLen + 1},
        \varUpsilon_{n,\iota\spLen + 1}\right).
\end{eqnarray}

\section{Solution Algorithm for Problem \pref{Reference Policy}} %
\label{section: solution of offline pre-allocation}
Problem \pref{Reference Policy} is a mixed-integer programming due to the BS association. Even without the optimization of BS association, the problem remains non-convex. In this section, an alternating optimization method is proposed to solve problem \pref{Reference Policy} iteratively, where $(\mathcal{E}_{\iota },\mathcal{T}_{\iota })$ and $\Upsilon_{\iota }$ are optimized alternatively in each iteration.

Denote the BS association, uplink time, and throughput allocation in the $i$-th iteration as $\mathcal{E}_{\iota }^i$, $\mathcal{T}_{\iota }^i$ and $\Upsilon_{\iota }^{i}$, respectively.
In each iteration (say the $i$-th iteration), the BS association $\mathcal{E}_{\iota }^i$ and time allocation $\mathcal{T}_{\iota }^i$ are first jointly optimized, given the throughput allocation from the previous iteration $\Upsilon_{\iota }^{i-1}$. 
Then, the throughput allocation $\Upsilon_{\iota }^{i}$ is optimized given the BS association $\mathcal{E}_{\iota }^i$ and time allocation $\mathcal{T}_{\iota }^i$. Both optimizations are elaborated below.

\subsubsection{Optimization of the BS association and time allocation}
Given throughput allocation  $\Upsilon_{\iota }^{i-1}$, the optimization of the BS association and time allocation can be written as
\begin{equation*}
	\begin{aligned}
		\plabel{Reference Policy - mixed integer BS time allocation}: \min_{ 
			\substack{ 
				\mathcal{E}_{\iota }, \mathcal{T}_{\iota }
			}}  \quad  & \omega_1 \sum_{t=(\iota-1)\spLen + 1  }^{T} \sum_{n \in \mathcal{N}} \sum_{m \in \mathcal{M}} 2^{\frac{ e_{n,m,t} r_{n,t}^{i-1}}{\txTime_{n,t} \slotLen \bandwidth}} \Phi_{n, m,t} \txTime_{n,t}, \\
		\text{s.t.} \quad & \txTime_{n,t} \ge 0,\forall n, t, \\
		&   e_{n,m,t} \in \{0, 1\},\forall n, m, t, \\
		& \sum_{n \in \mathcal{N}} \txTime_{n,t} e_{n,m,t} \le 1, \forall m,t, \\
		& \sum_{m \in \mathcal{M}} e_{n,m,t} = 1, \forall n,t, \\
		& 2^{\frac{ e_{n,m,t} r_{n,t}^{i-1}}{\txTime_{n,t} \slotLen \bandwidth}} \Phi_{n, m,t} \le P_{\max}, \forall n,m,t
	\end{aligned}
\end{equation*}
where $r_{n,t}^{i-1}$ is the throughput allocation for the $n$-th vehicle in $\Upsilon_{\iota }^{i-1}$, and 
$$\Phi_{n,m,t} \triangleq 2^{-\mathds{E}_h \left[ \log_2 \frac{g_{n,m,t}\vert h \vert^2}{\sigma^2} \right]}.$$ 
The transmission time cost and data cost in the objective of \pref{Reference Policy} are omitted since they are constants. The third constraint ensures that the total transmission time to each BS is limited to one time slot, while the fourth constraint guarantees that each vehicle is associated with only one BS per time slot.

The above optimization problem is non-convex, it can be solved by relaxing the integer constraints: each vehicle can transmit to multiple BSs in one time slot. Particularly, we relax the binary constraint of $e_{n,m,t}\in \{0,1\}$ to $e_{n,m,t} \in [0, 1]$, where $e_{n,m,t}$ after relaxation represents the data fraction transmitted from the $n$-th vehicle to the $m$-th BS in the $t$-th time slot. Accordingly, the uplink time variables are extended from $\mathcal{T}_{ \iota  }$ to
$$
\hat{\mathcal{T}}_{ \iota  } \triangleq \left( \hat{\boldsymbol{\txTime}}_{(\iota-1)\spLen + 1 }, \ldots, \hat{\boldsymbol{\txTime}}_T \right),
$$
where $\hat{\boldsymbol{\txTime}}_{t} \in \mathbb{R}^{N \times M}$ is a time allocation matrix, and its element $\hat{\txTime}_{n,m,t} = [\hat{\boldsymbol{\txTime}}_{\iota }]_{n,m}$ denotes the uplink time of the $n$-th vehicle to the $m$-th BS in the $t$-th time slot.
The relaxed optimization problem can then be written as
\begin{equation*}
	\begin{aligned}
		\plabel{Reference Policy - BS time allocation}: \min_{ 
			\substack{ 
				\mathcal{E}_{\iota }, \hat{\mathcal{T}}_{ \iota  }
			}}  \quad  & \omega_1 \sum_{t=(\iota-1)\spLen + 1  }^{T} \sum_{n \in \mathcal{N}} \sum_{m \in \mathcal{M}} 2^{\frac{ e_{n,m,t} r_{n,t}^{i-1}}{\hat{\txTime}_{n,m,t} \slotLen \bandwidth}} \Phi_{n, m,t} \hat{\txTime}_{n,m,t} , \\
		\text{s.t.} \quad &  \hat{\txTime}_{n, m, t} \ge 0,  e_{n, m,t} \ge 0,\forall n, m, t, \\
		& \sum_{n \in \mathcal{N}} \hat{\txTime}_{n,m,t} \le 1, \forall m,t, \\
		& \sum_{m \in \mathcal{M}} e_{n,m,t} = 1, \forall n,t. \\
		& 2^{\frac{ e_{n,m,t} r_{n,t}^{i-1}}{\hat{\txTime}_{n,m,t} \slotLen \bandwidth}} \Phi_{n, m,t} \le P_{\max}, \forall n,m,t.
	\end{aligned}
\end{equation*}
Note that the above problem \pref{Reference Policy - BS time allocation} is convex, it can be solved efficiently. Let $\{e_{n,m,t}^{\ast}|\forall m\}$ be the optimal BS association of problem \pref{Reference Policy - BS time allocation}. The BS association for problem \pref{Reference Policy - mixed integer BS time allocation} can be determined as 
\begin{equation}
    e_{n,m,t}^{i} = \mathds{1} ( m = \argmax_{k \in \mathcal{M}} e_{n,k,t}^{\ast} ), \forall n,t.
\end{equation}
Then, by solving the problem \pref{Reference Policy - mixed integer BS time allocation} with $e_{n,m,t}=e_{n,m,t}^{i}$, $\forall n,m,t$, which becomes convex, the optimal time allocation $\txTime_{n,t}^{i}$, $\forall n,t$, can be determined. Aggregating $e_{n,m,t}^{i}$ and $\txTime_{n,t}^{i}$, the overall BS association $\mathcal{E}_{\iota }^i$ and time allocation $\mathcal{T}_{\iota }^i$ can be obtained, respectively.

\subsubsection{Optimization of throughput}
Given the optimized BS associations $ \mathcal{E}_{\iota }^i $ and time allocations $\mathcal{T}_{\iota }^i$, the throughput allocation of each vehicle in \pref{Reference Policy} can be decoupled.
The optimization of throughput allocation of the $n$-th vehicle in all the remaining time slots can be written as
\begin{equation*}
	\begin{aligned}
		\plabel{Reference Policy - Throughput Optimization}: \min_{ \varUpsilon_{n,(\iota-1)\spLen + 1} }  \quad  & \sum_{t=(\iota-1)\spLen + 1  }^{T} \mathsf{c}_{n,t} \left(
				d_{n,t}, \boldsymbol{l}_{n,t}^{(\spIdx)}, [\boldsymbol{E}_t^{i}]_n, 
				\txTime_{n,t}^{i}, r_{ n,t}
		\right), \\
		\text{s.t.} \quad & \cref{BS Association}, \cref{Uplink Time Constraint 1}, \cref{Uplink Time Constraint 2}, \cref{Peak Power Constraint},
	\end{aligned}
\end{equation*}
where $\varUpsilon_{n,(\iota-1)\spLen + 1}$ is defined in \eqref{symbol: throughput allocation}.
The optimization of problem \pref{Reference Policy - Throughput Optimization} can be conducted in two cases: the task offloading can be accomplished in the scheduling period and otherwise. Both cases are discussed below.

\textbf{Case 1}: Suppose the task offloading of the $n$-th vehicle is accomplished in the $\nthVehicleFinishTime$-th time slot and  $\nthVehicleFinishTime \leq T$, the problem P3 can be rewritten as the following problem \pref{Reference Policy - Throughput Optimization - Case 1}. 
\begin{equation*}
	\begin{aligned}
		\plabel{Reference Policy - Throughput Optimization - Case 1}: \min_{ r_{n,(\iota-1)\spLen + 1}, ..., r_{n,\nthVehicleFinishTime} } \quad & \omega_1 \sum_{t = (\iota-1)\spLen + 1  }^{\nthVehicleFinishTime} 2^{\frac{r_{n,t}}{\txTime_{n, t}^{i}\slotLen \bandwidth }} \Phi_{n,t}  \txTime_{n, t}^{i}, \\
		\text{s.t.}  \quad
		& 0 \le r_{n,t} \le \txTime_{n, t}^{i}\slotLen \bandwidth \log_2 \frac{P_{\max}}{\Phi_{n,t}}, \forall t, \\
		& \sum_{t=(\iota-1)\spLen + 1}^{\nthVehicleFinishTime} r_{n,t} = D_{n,a}, \\
	\end{aligned}
\end{equation*}

\textbf{Case 2}: If the task offloading cannot be accomplished, the optimization problem can be written as
\begin{equation*}
	\begin{aligned}
		\plabel{Reference Policy - Throughput Optimization - Case 2}: \min_{ \varUpsilon_{n,(\iota-1)\spLen + 1}}\quad & \omega_1 \sum_{t = (\iota-1)\spLen + 1  }^{T} 2^{\frac{r_{n,t}}{\txTime_{n, t}^{i}\slotLen \bandwidth}} \Phi_{n,t}  \txTime_{n, t}^{i}  \\ 
		& +  \omega_2 \left(D_n^a - \sum_{t = (\iota-1)\spLen + 1  }^{T} r_{n,t}\right), \\
		\text{s.t.}  \quad
		& 0 \le r_{n,t} \le \txTime_{n, t}^{i}\slotLen \bandwidth \log_2 \frac{P_{\max}}{\Phi_{n,t}}, \forall t, \\
		& \sum_{t = (\iota-1)\spLen + 1  }^{T} r_{n,t} \le D_{n,a}.  \\
	\end{aligned}
\end{equation*}
The solution for the problem \pref{Reference Policy - Throughput Optimization - Case 1} and \pref{Reference Policy - Throughput Optimization - Case 2} are provided in Appendix \ref{Appendix: offline pre-allocation}. 
After solving them, the optimal throughput allocation for \pref{Reference Policy - Throughput Optimization} is determined by comparing the objective values of \pref{Reference Policy - Throughput Optimization - Case 1} and \pref{Reference Policy - Throughput Optimization - Case 2}, which is denoted as $\Upsilon_{\iota }^i$.

\section{Optimization of the Online Scheduling \pref{Online Optimization Problem}}  
\label{section: solution of dynamic improvement}

In this section, an optimization algorithm for the online scheduling problem \pref{Online Optimization Problem} is proposed. It can be observed that the problem is convex when $t = T$, representing the final time slot of the scheduling period. However, for $t < T$, the problem becomes non-convex due to the approximation of the value function. The former can be solved with existing method efficiently. To address the latter, an alternating minimization method can be applied.

The variables of problem \pref{Online Optimization Problem} include the BS association $\boldsymbol{E}_t$, time allocation $\boldsymbol{\txTime}_t$ and throughput allocation $\boldsymbol{r}_t$ for $t$-th time slot, as well as the reference throughput allocation $\{\varUpsilon_{n,t+1} | n \in \mathcal{N} \}$ for the value function approximation. In the proposed alternating minimization method, $(\boldsymbol{E}_t,\boldsymbol{\txTime}_t)$ and $(\boldsymbol{r}_t,\{\varUpsilon_{n,t+1} | n \in \mathcal{N} \})$ are optimized alternatively in each iteration.
Particularly, denote the BS association,  time and throughput allocation for $t$-th time slot in the $i$-th iteration as $\boldsymbol{E}_{t}^i$, $\boldsymbol{\txTime}_{t}^i$ and $\boldsymbol{r}_{t}^{i}$, respectively, and denote the reference throughput allocation for value function approximation in the $i$-th iteration as $\{\varUpsilon_{n,t+1}^{i} | n \in \mathcal{N} \}$.
In each iteration (say the $i$-th iteration), we jointly optimize the BS association $\boldsymbol{E}_t^i$ and time allocation $\boldsymbol{\txTime}_t^i$ with the throughput allocation $\boldsymbol{r}_t^{i-1}$ and $\{\varUpsilon_{n,t+1}^{i-1} | n \in \mathcal{N} \}$. 
Next, for each vehicle (say the $n$-th one, $n \in \mathcal{N}$), we optimize its throughput allocation $r_{n,t}^{i}$ and the reference throughput allocation $\varUpsilon_{n,t+1}^{i}$, given the previously optimized BS association $\boldsymbol{E}_{t}^i$ and time allocation $\boldsymbol{\txTime}_{t}^i$. Both two steps are elaborated below.

\subsubsection{Optimization of the BS association and time allocation} 
Note that the optimal solution of BS association should search all the possible combinations, leading to prohibitive complexity. In order to search a sub-optimal association with low complexity, we first relax the problem by allowing each vehicle to associate with all the BSs in one time slot (the same relaxation method as introduced in Section~\ref{section: offline pre-allocation}). 

Particularly, let $e_{n,m,t} \in [0, 1]$ be the data fraction transmitted from the $n$-th vehicle to the $m$-th BS in the $t$-th time slot, and $\hat{\txTime}_{n,m,t} = [\hat{\boldsymbol{\txTime}}_{\iota }]_{n,m}$ be the uplink time of the $n$-th vehicle to the $m$-th BS in the $t$-th time slot. Hence, given the throughput allocation, the optimization of the BS association and time allocation with relaxation can be written as
\begin{equation*}
    \begin{aligned}
    \plabel{Online Optimization - BS time allocation}: \min_{\boldsymbol{E}_{t}, \hat{\boldsymbol{\txTime}}_{t}} \quad
                           & \omega_1 \sum_{n \in \mathcal{N}} \sum_{m \in \mathcal{M}} 2^{\frac{e_{n, m,t} r_{n,t}^{i-1}}{\hat{\txTime}_{n,m,t}\slotLen \bandwidth}} \Phi_{n,t} \hat{\txTime}_{n,m,t}, \\
    \text{s.t.}  \quad &  \hat{\txTime}_{n, m, t} \ge 0, e_{n, m,t} \ge 0,\forall n, m, \\
            & \sum_{n \in \mathcal{N}} \hat{\txTime}_{n,m,t} \le 1, \forall m, \\
            & \sum_{m \in \mathcal{M}} e_{n,m,t} = 1, \forall n, \\
            & 2^{\frac{ e_{n,m,t} r_{n,t}^{i-1}}{\hat{\txTime}_{n,m,t} \slotLen \bandwidth}} \Phi_{n, m,t} \le P_{\max}, \forall n,m.
    \end{aligned}
\end{equation*}
Note that the above problem \pref{Online Optimization - BS time allocation} is convex, which can be solved efficiently with the existing method. Let $\{e_{n,m,t}^{\ast}|\forall m,n\}$ be the optimal solution of $\{e_{n,m,t}|\forall m,n\}$. The BS association can be determined as 
\begin{equation}
    e_{n,m,t}^{i} = \mathds{1} ( m = \argmax_{k \in \mathcal{M}} e_{n,k,t}^\ast ).
\end{equation}
Then, the optimal time allocation with the above BS association can be obtained by solving the convex problem \pref{Online Optimization - BS time allocation} with $e_{n,m,t}=e_{n,m,t}^{i}$, $\forall n,m$.
Let $\{ \txTime_{n,m,t}^\ast | \forall m \}$ be the optimal time allocation of $\{ \hat{\txTime}_{n,m,t} | \forall m \}$, $\forall n$,  the optimal time allocation $\txTime_{n,t}^{i}$ is determined from 
\begin{equation}
	\txTime_{n,t}^{i} = \sum_{m \in \mathcal{M}} \txTime_{n,m,t}^{\ast} \mathds{1} ( m = \argmax_{k \in \mathcal{M}} e_{n,k,t}^\ast ).
\end{equation} 

\begin{figure}[t]
    \centering

    \begin{subfigure}{0.45\linewidth}
    \centering
	\includegraphics[width=\linewidth]{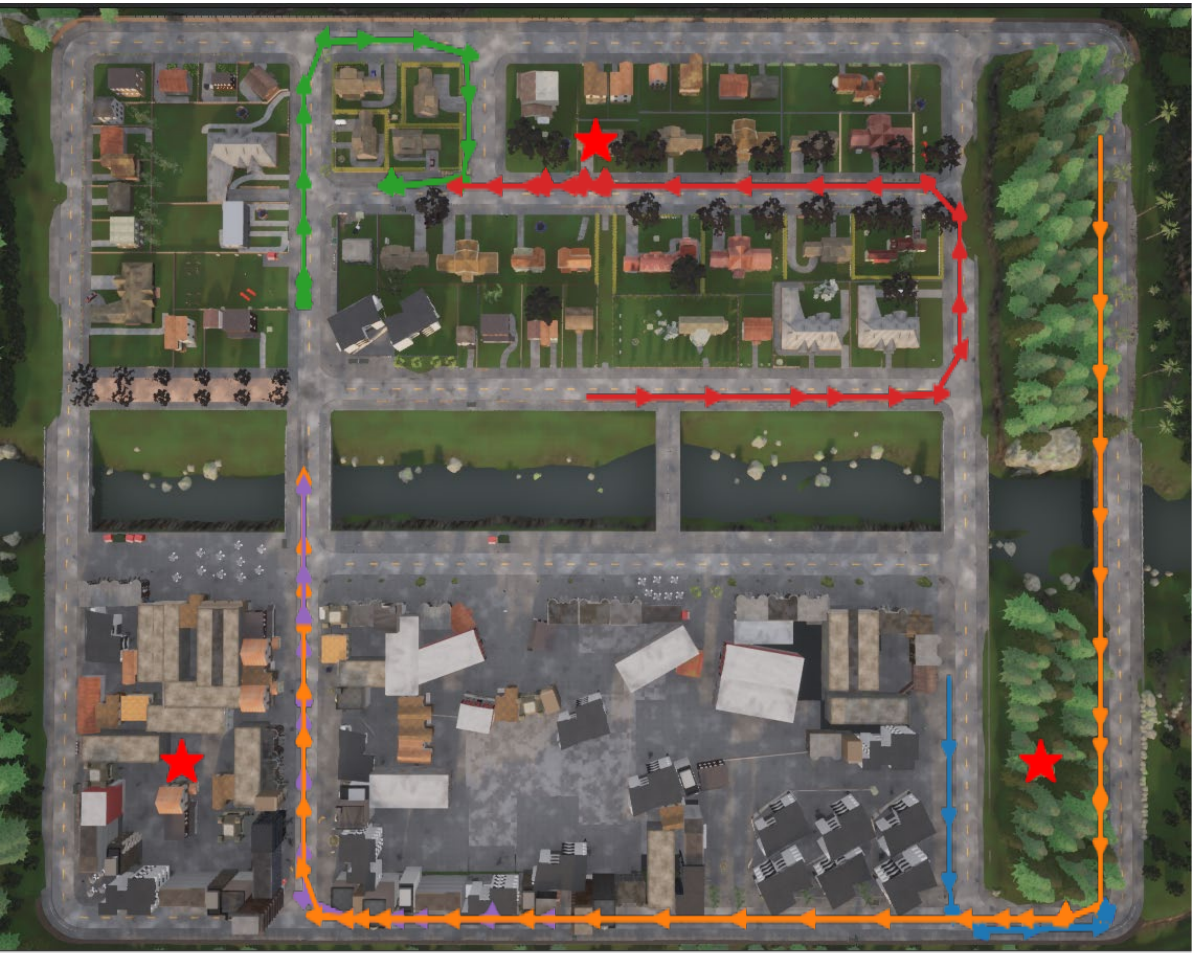}
	\end{subfigure}
    \begin{subfigure}{0.45\linewidth}
    \centering
	\includegraphics[width=\linewidth]{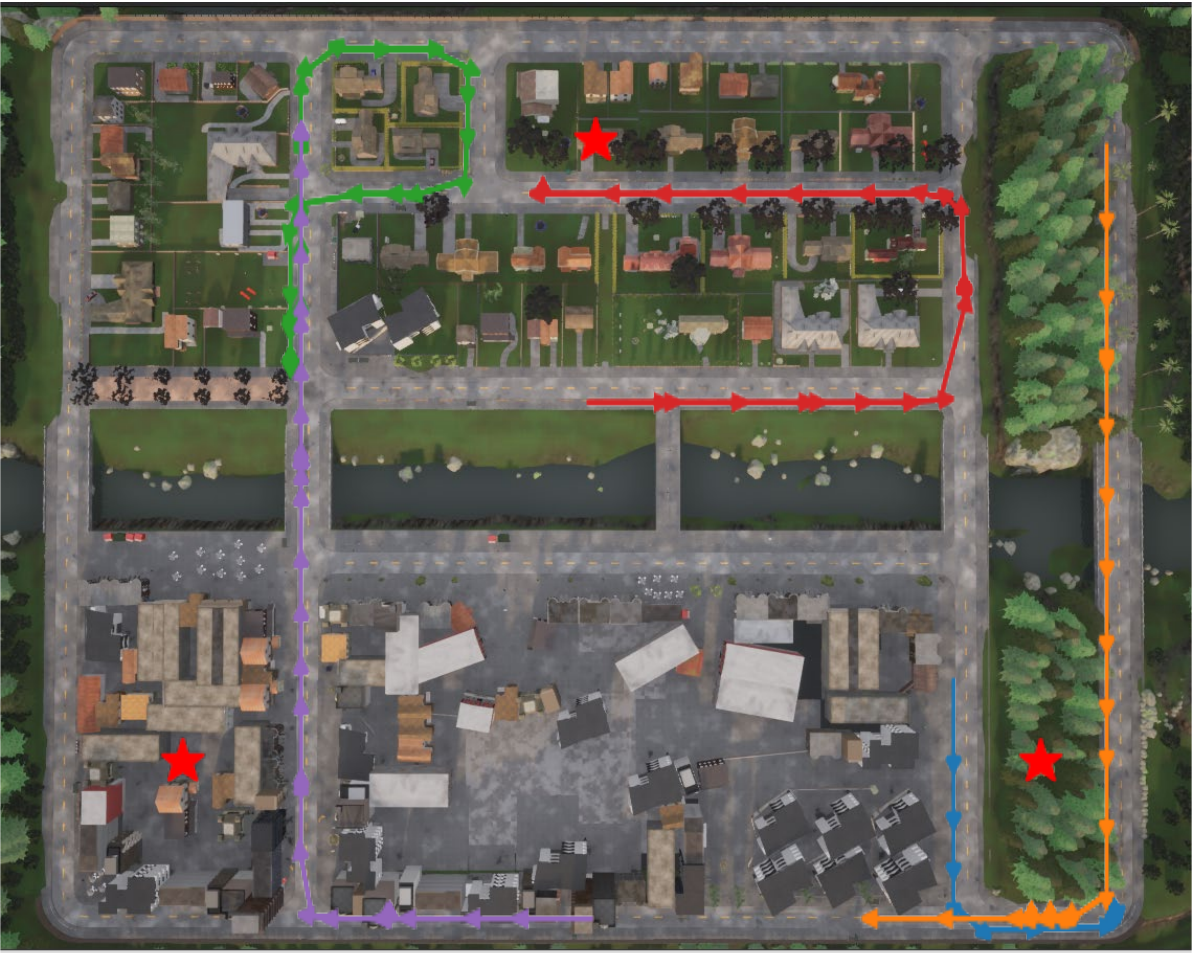}
	\end{subfigure}
	
	\add{
	\vspace{0.15cm}
    \begin{minipage}{\linewidth}
    \includegraphics[width=\textwidth]{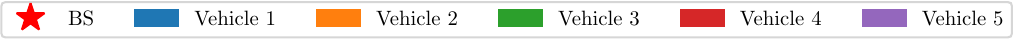}
    \end{minipage}
    }

	\caption{Illustration of simulation scenario, where two samples of random trajectories are provided. The stars and lines denote BSs' locations and vehicles' trajectories, respectively. Due to random velocities, the lengths of trajectories of each vehicle could be different in the two samples.}
	\label{figure: illustration of simulator}
\end{figure}

\subsubsection{Optimization of the throughput}
Given the BS association and time allocation, the throughput optimization of $N$ vehicles can be decoupled, where the sub-problem for the $n$-th vehicle can be written as
\begin{equation*}
    \begin{aligned}
    \plabel{Online Optimization - throughput allocation}: \min_{ \substack{r_{n,t}, \\ \varUpsilon_{n,t+1} }} \quad & 
    \omega_1 2^{\frac{r_{n,t}}{\txTime_{n,t}^i\slotLen \bandwidth}}\Phi_{n,t}\txTime_{n,t}^i \\
    & + f_{n, t+1}^1(d_{n,t} - r_{n,t}, \varUpsilon_{n,t+1}) \\ 
    & + \omega_1 \sum_{ \boldsymbol{l}_{n,t+1} } \mathds{P}\left[ \boldsymbol{l}_{n,t+1}|  \boldsymbol{l}_{n,t} \right] \\ 
    & \times f_{n, t+1}^2(d_{n,t} - r_{n,t}, \boldsymbol{l}_{n,t+1}, \varUpsilon_{n,t+1}), \\
    \text{s.t.}  \quad
    & 0 \le r_{n,t} \le \txTime_{n, t}^i \slotLen \bandwidth \log_2 \frac{P_{\max}}{\Phi_{n,t}}. \\ 
    \end{aligned}
\end{equation*}
The above problem \pref{Online Optimization - throughput allocation} can be solved similarly to the  problem \pref{Reference Policy - Throughput Optimization}.

\section{Complexity and Performance Analysis}
\label{section: complexity and performance analysis}
The complexity of the proposed algorithm is analyzed as follows. The overall scheduling algorithm consists of two components: super-slot-scale per-allocation and slot-scale online scheduling. The former involves iterative optimization of problems \pref{Reference Policy - mixed integer BS time allocation} and \pref{Reference Policy - Throughput Optimization}. The computational complexity of the problem in \pref{Reference Policy - mixed integer BS time allocation} is $O((NT)^2/\epsilon^2)$ to achieve an $\epsilon$-accurate solution. This results from the epigraph form of \pref{Reference Policy - mixed integer BS time allocation}, where the maximum feasible set has an upper bound of $P_{\max}NT$, and the objective is $1$-Lipschitz continuous. Using the projected gradient method \cite[Thm. 8.13]{beckFirstOrderMethodsOptimization2017}, obtaining an $\epsilon$-accurate solution requires $O((NT)^2/\epsilon^2)$ iterations.
The computational complexity of problem \pref{Reference Policy - Throughput Optimization} is $O(NT^2\log_2 T)$. Consequently, the total complexity per super-slot is 
$$O\left(I_{\max}\left((NT)^2 /\epsilon^2 + NT^2\log_2 T\right)\right),$$ 
where $I_{\max}$ is the maximum number of iterations.

The slot-scale online scheduling consists of iterative optimization of problems \pref{Online Optimization - BS time allocation} and \pref{Online Optimization - throughput allocation}. Their complexities are $O(N^2/\epsilon^2)$ and $O(NT^2 \log_2 T)$ respectively. Thus, the overall complexity of online scheduling is 
$$O\left( I_{\max} \left( N^2/\epsilon^2 + NT^2\log_2 T \right)\right),$$ 
where $I_{\max}$ is the maximum iteration number.

Note that both pre-allocation and online scheduling are designed to suppress the average total cost of the reference scheduling, the performance of the proposed scheduling framework can actually be bounded by the reference scheduling, which is summarized below.
\begin{Lemma}[Non-trivial Performance Bound]
Let $\widetilde{\Omega}$ be the proposed scheduling policy. With sufficiently large peak power constraint, the average total cost of the proposed scheduling framework since the $\iota$-th super slot ($\forall \iota$) is bounded above by that of the corresponding reference scheduling. Specifically,
\begin{equation}
	 \begin{aligned}
	\mathsf{C}_t( \state_t, \widetilde{\Omega}) \leq\  &c_t(\state_t, \action_{t,\iota}^{ref}) +  \sum_{\state_{t+1}} \mathds{P}\left[ \state_{t+1} | \state_{t}, \action_t \right] \\
	&\times \!\sum_{n \in \mathcal{N}} \widetilde{V}_{n,t+1}\left(
        d_{n,t+1},\boldsymbol{l}_{n,t+1}, \varUpsilon_{n,t+1}^{(\iota)}\right), \forall \iota, t,
	\end{aligned}
\end{equation}
where 
$$\varUpsilon_{n,t+1}^{(\iota)} \triangleq \left(r_{n,t+1}^{(\iota)}, r_{n,t+1}^{(\iota)}, ..., r_{n,T}^{(\iota)}\right),$$ 
denotes the aggregation of throughput allocation in the reference scheduling, and $\Tilde{V}_{n,t+1}(\cdot)$ is defined in {\rm Lemma \ref{lemma: Asymptotically Achievable Average Cost}}.
\end{Lemma}
\begin{proof}
This result follows directly from the fact that the online scheduling in each time slot successively suppresses the average total cost of the remaining time slots based on the reference scheduling $\mathcal{R}_\iota$.
\end{proof}

\begin{Remark}[Analysable Approximate MDP]
Most of the approximate MDP methods are hard to analyze. In this paper, we incorporate physical meaning to the approximation of the optimal value function, such that an upper bound on the average total cost can be obtained in Lemma 2. The intuitive explanation is given below:
\begin{itemize}
    \item {Resource allocation with physical meaning:} we assume each vehicle moves along a deterministic trajectory, i.e., average trajectory, in the future, such that the resource allocation for the future time slots can be optimized via deterministic optimization in advance. Although it is obviously not optimal, it is a reasonable allocation on average. As a remark, note that such average trajectories are updated in each super slot, keeping them up to date. 
    \item {Value function with physical meaning:} given the above deterministic resource allocation, we derive the analytical expression of average future cost by considering the statistics of random trajectories, and use this expression as the approximation of optimal value function. Clearly, such approximation represents an achievable average future cost. Therefore, the online scheduling with such approximated value function performs no worse than the value function itself.
\end{itemize}
\end{Remark}

\begin{Remark}[Extremely Unpredictable Trajectories]
Generally, the performance gain of the proposed algorithm will decay when the mobility is extremely unpredictable, due to the inherent property of MDP. It can be observed from the Bellman's equation that in each time slot, the MDP-based scheduler jointly optimizes (minimizes) the system cost of the current time slot plus the average system cost in the future. The latter is represented by the value function in the Bellman's equation. Hence, one feature of MDP is that the scheduler not only considers the action impact on the current time slot, but also its impact on the future. If the mobility becomes extremely unpredictable, the above feature of MDP may become less important, leading to less performance gain. This phenomenon is captured analytically by the upper bound derived in Lemma 2. Particularly, it can be proved that the upper bound increases monotonically with a factor $\mathds{E} \left[ \Vert \boldsymbol{l}_{n,t+\kappa} - \boldsymbol{l}_m \Vert_2^{\gamma} \mid \boldsymbol{l}_{n,t} \right]$, which increases with mobility uncertainty.
\end{Remark}

\section{Simulation and Discussion}
\label{section: Performance Evaluation}

\begin{figure*}[t]
    \centering
    \begin{minipage}{\textwidth}
    \centering
    \includegraphics[width=0.9\textwidth]{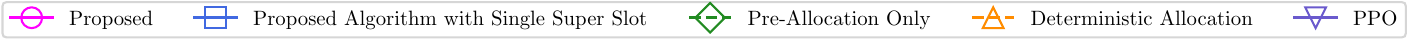}
    \end{minipage}
    
    \vspace{0.15cm}
    
    \begin{subfigure}{0.24\linewidth}
        \centering
        \includegraphics[width=\linewidth]{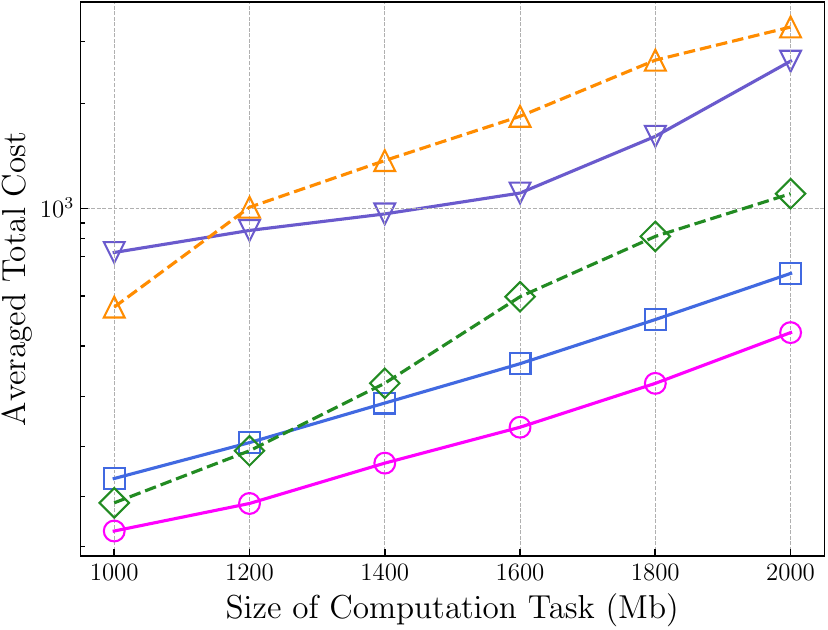}
        \caption{}
        \label{figure: cost comparison under different data arrival - total cost}
    \end{subfigure}
    \begin{subfigure}{0.24\linewidth}
        \centering
        \includegraphics[width=\linewidth]{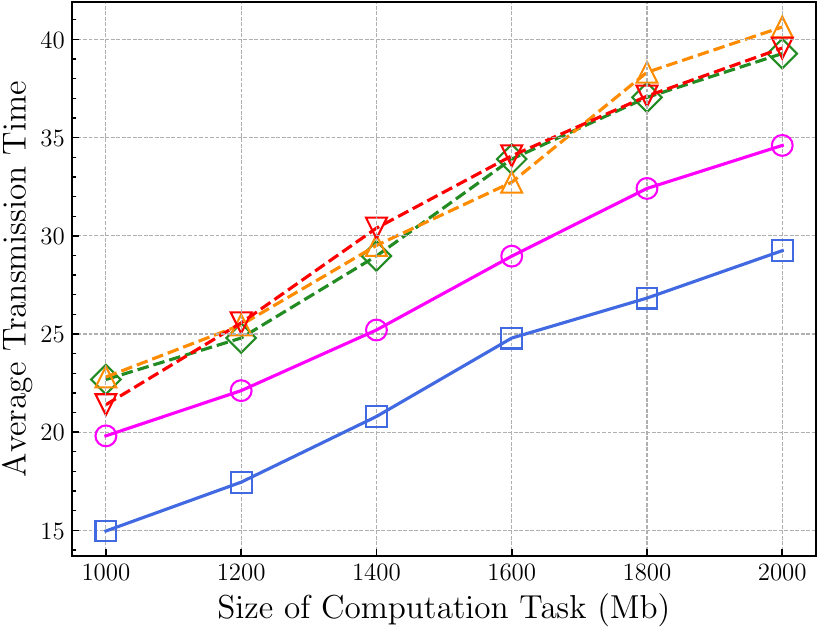}
        \caption{}
        \label{figure: cost comparison under different data arrival - time cost}
    \end{subfigure}
    \begin{subfigure}{0.24\linewidth}
        \centering
        \includegraphics[width=\linewidth]{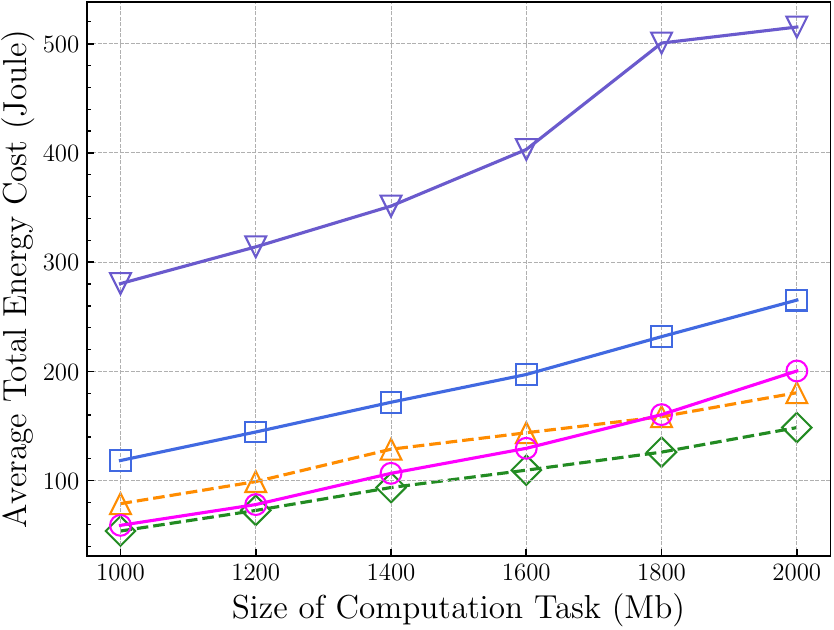}
        \caption{}
        \label{figure: cost comparison under different data arrival - energy cost}
    \end{subfigure}
    \begin{subfigure}{0.24\linewidth}
        \centering
        \includegraphics[width=\linewidth]{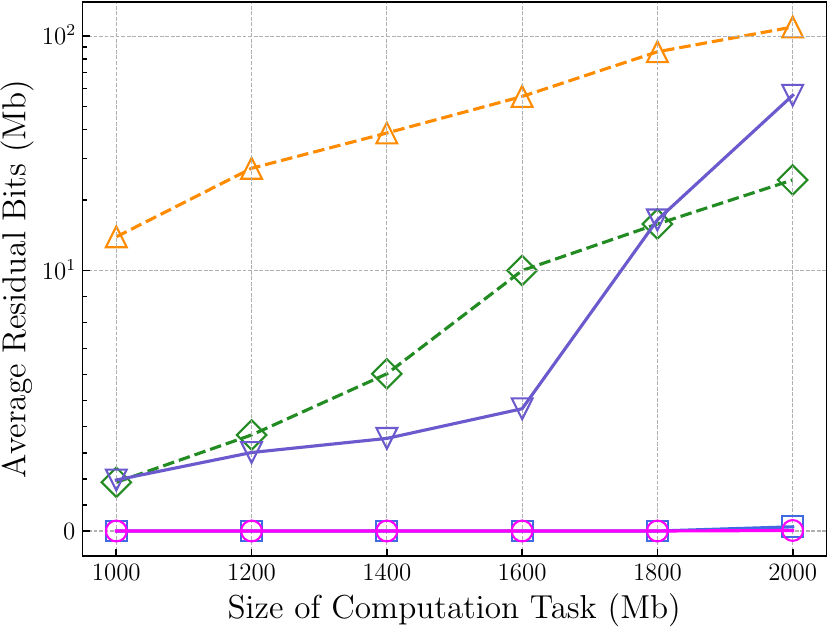}
        \caption{}
        \label{figure: cost comparison under different data arrival - data cost}
    \end{subfigure}
    \caption{The comparison of (a) average total cost, (b) average transmission time, (c) average energy cost and (d) average residual information bits per vehicle versus different sizes of computation task.}
    \label{figure: cost comparison under different data arrival}
\end{figure*}

In this section, the performance of the proposed solution framework is evaluated via numerical simulation.
The vehicle's trajectories are randomly generated from the high-fidelity CARLA ({Car} {L}earning to {A}ct) simulator \cite{dosovitskiy2017carla} in the road map of \textit{Town01}, which has physical dimensions $410.68\ \mathrm{m} \times 344.26\ \mathrm{m}$. 
As shown in \figurename~\ref{figure: illustration of simulator}, the simulation involves $N = 5$ vehicles and $M = 3$ BSs. 
The scheduling period spans $\periodLen = 50$ seconds, divided into $50$ time slots. Hence, each time slot is with a duration of $\slotLen = 1$ second.
Each BS operates on a bandwidth of $\bandwidth = 20$ MHz. The path loss exponent is $\gamma = 4$, and the total power of noise and interference is $-30$ dBw.
Information bits arrive at each vehicle at the beginning of the scheduling period. The peak power constraint is $P_{\max} = 5$ W. The transmission signal experiences Rayleigh fast fading, where $h\sim\mathcal{CN}(0,6)$. The weights $\omega_1$ and $\omega_2$ are $2$ and $5$, respectively, the number of information bits per task is $1600$M, each super slot consists of $10$ time slots. These parameters are fixed in all the simulations unless otherwise stated.

In the simulation, we evaluate the performance of the proposed framework and the following baseline algorithms. Each algorithm is simulated for 500 trails, such that the average total cost can be obtained. 

\begin{itemize}
  \item \textbf{Deterministic Allocation (DA)}: The reference scheduling of the first super slot $\mathcal{R}_1$ is directly utilized for offloading scheduling.
  \item \textbf{Pre-Allocation Only (PAO)}: The reference scheduling is updated and directly utilized for offloading scheduling in each super slot.
  \item \textbf{Proposed Algorithm with Single Super Slot (PAS3)}: In the proposed algorithm, the super slot spans the whole scheduling period. Thus, the reference scheduling is calculated only at the beginning of the scheduling period, similar to our prior work \cite{li2025dynamic}.
  \item \textbf{Proximal Policy Optimization (PPO)}: Conventional reinforcement learning algorithm for MDP as elaborated in \cite{schulman2017proximal}. 
Particularly, the PPO algorithm employs multi-layer perceptron networks for both actor and critic components. Each network features a single hidden layer containing 256 units. The training procedure comprises up to 156,250 episodes. Within each episode, the agent completes 64 discrete cycles. Every cycle consists of 50 consecutive time slots of agent-environment interaction, after which the environment resets. The discount factor is set to 0.
The learning rates are set to $3\times 10^{-4}$ for the actor network and $1\times 10^{-3}$ for the critic network.
\end{itemize}

\begin{figure*}[t]
    \centering
    \begin{minipage}{\textwidth}
    \centering
    \includegraphics[width=0.9\textwidth]{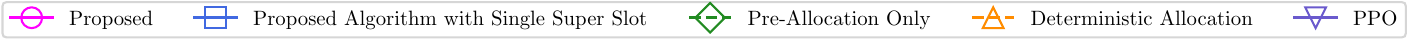}
    \end{minipage}
    
    \vspace{0.15cm}
    \begin{subfigure}{0.24\linewidth}
        \centering
        \includegraphics[width=\linewidth]{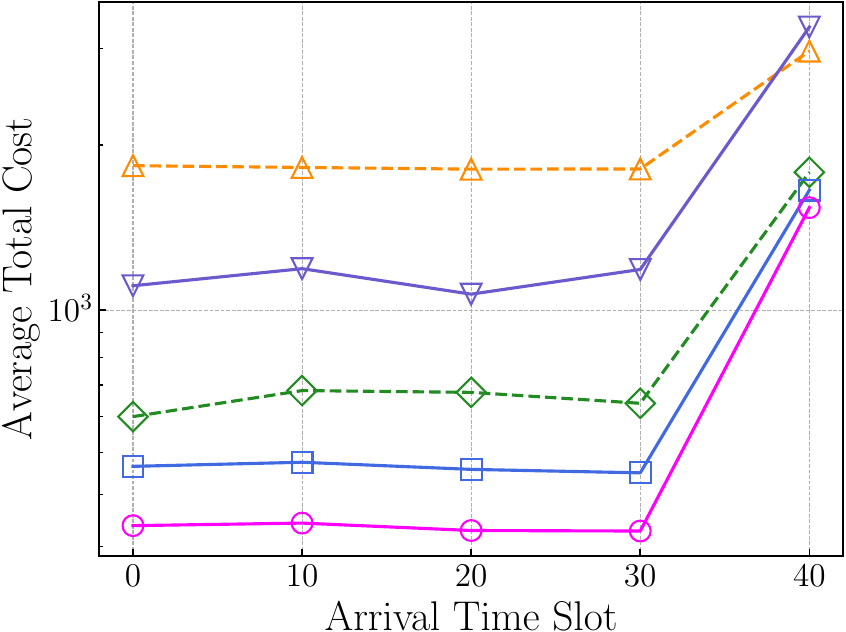}
        \caption{}
        \label{figure: cost comparison under different arrival time slot - total cost}
    \end{subfigure}
    \begin{subfigure}{0.24\linewidth}
        \centering
        \includegraphics[width=\linewidth]{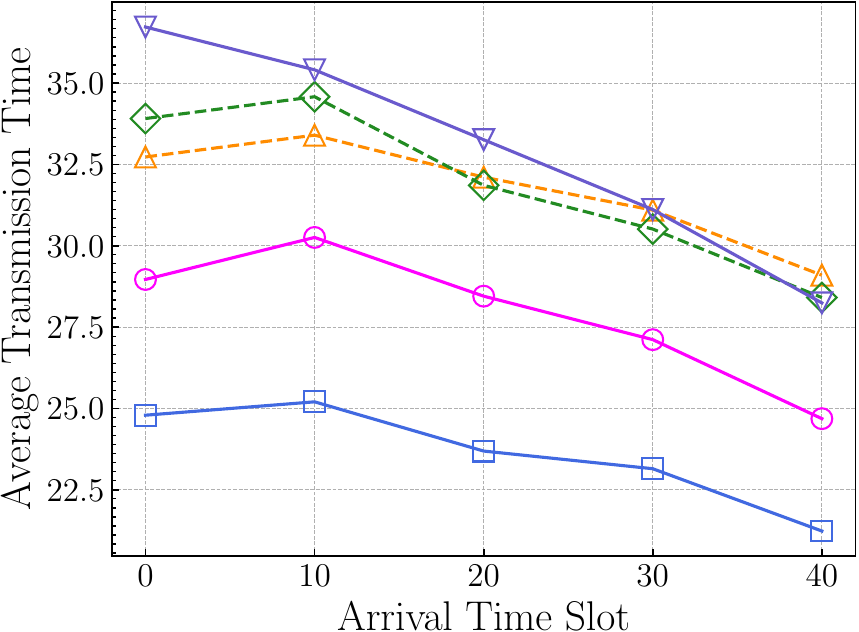}
        \caption{}
        \label{figure: cost comparison under different arrival time slot - time cost}
    \end{subfigure}
    \begin{subfigure}{0.24\linewidth}
        \centering
        \includegraphics[width=\linewidth]{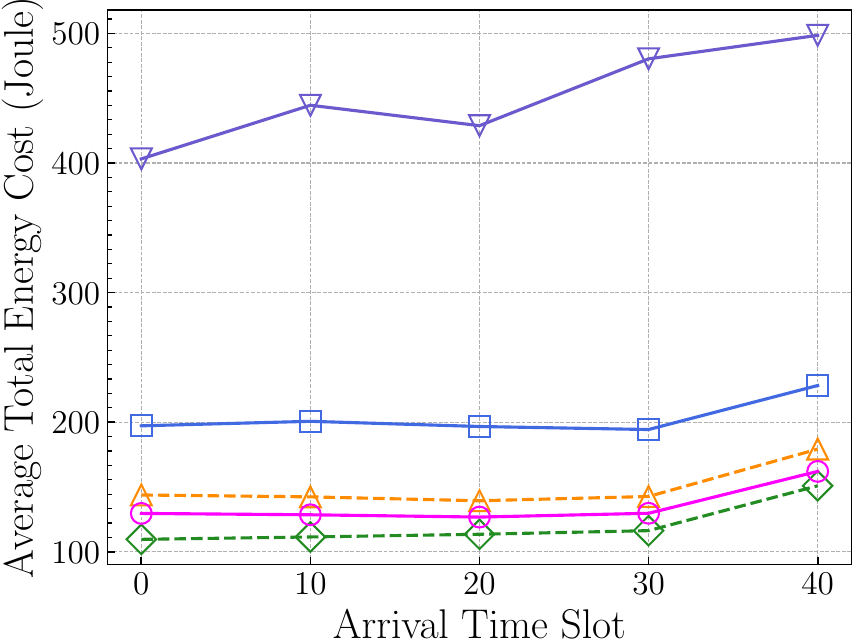}
        \caption{}
        \label{figure: cost comparison under different arrival time slot - energy cost}
    \end{subfigure}
    \begin{subfigure}{0.24\linewidth}
        \centering
        \includegraphics[width=\linewidth]{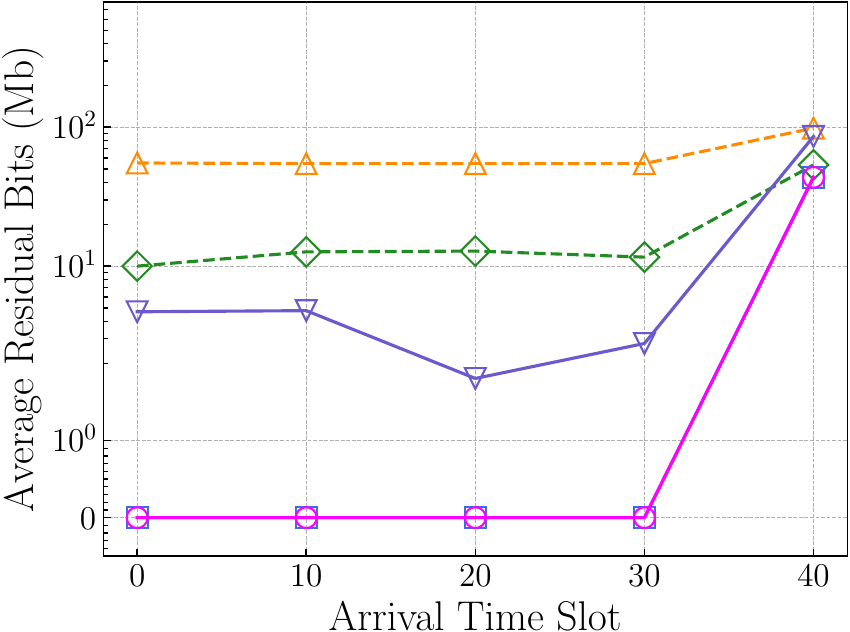}
        \caption{}
        \label{figure: cost comparison under different arrival time slot - data cost}
    \end{subfigure}
    \caption{The comparison of (a) average total cost, (b) average transmission time, (c) average energy cost and (d) average residual information bits per vehicle versus different task arrival time $T_{1,a}$.}
    \label{figure: cost comparison under different arrival time slot}
\end{figure*}

\begin{figure}[!t]
    \centering
    \vspace{1em}
    \begin{subfigure}{0.49\linewidth}
        \centering
        \includegraphics[width=\linewidth]{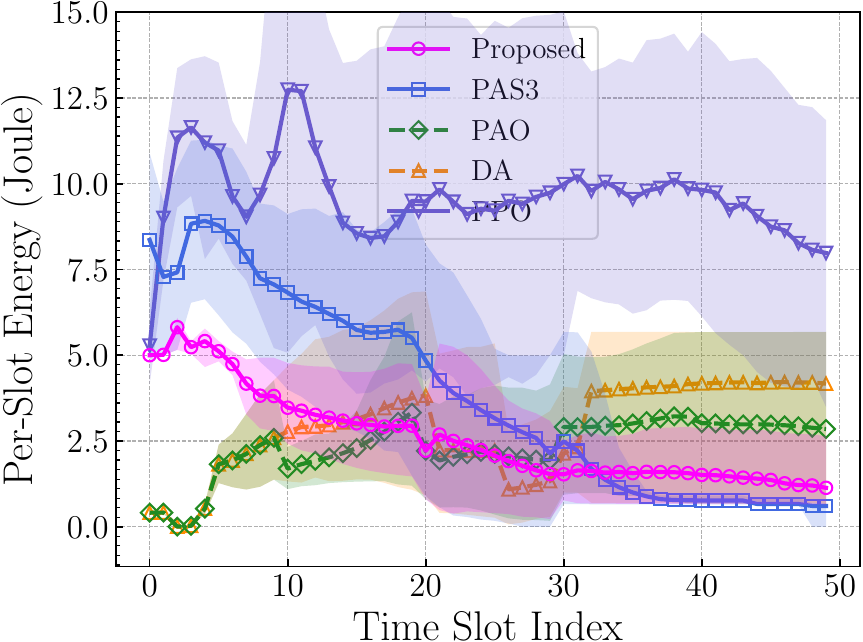}
        \caption{Task arrives in the $30$-th slot}
    \end{subfigure}
    \begin{subfigure}{0.49\linewidth}
        \centering
        \includegraphics[width=\linewidth]{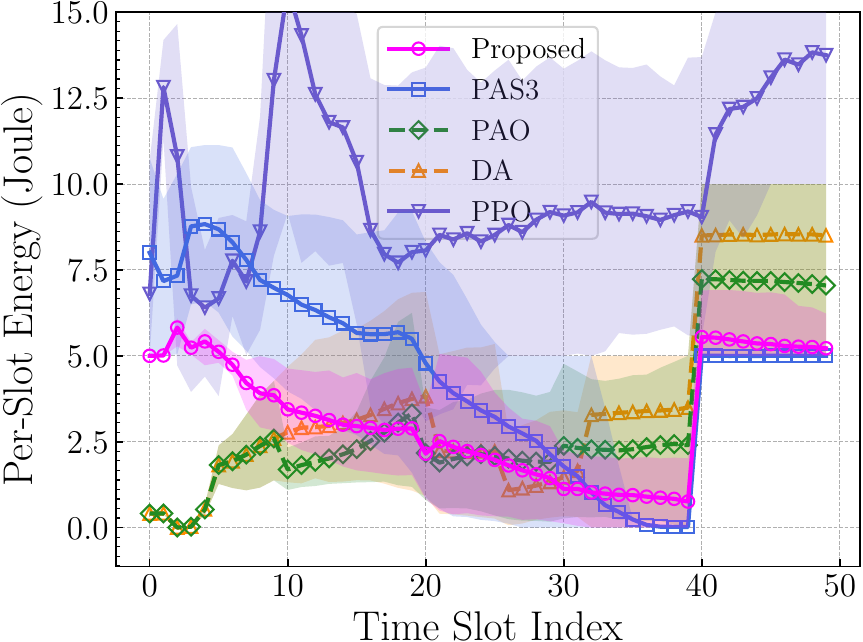}
        \caption{Task arrives in the $40$-th slot}
    \end{subfigure}
    \caption{Illustration of per-slot energy consumption.}
    \label{fig: per slot power dynamic arrival}
\end{figure}

\subsection{Performance versus Task Size}
\label{sec: analysis on the varying data arrival}
The average total costs of the proposed scheduling algorithm and the baselines are shown in \figurename~\ref{figure: cost comparison under different data arrival - total cost} with different sizes of computation task (ranging from $1000$Mb to $2000$Mb). Additional comparisons on average transmission time slots per vehicle, average total energy cost, and average residual information bits per car at the end of scheduling period are demonstrated in \figurename~\ref{figure: cost comparison under different data arrival - time cost}, \figurename~\ref{figure: cost comparison under different data arrival - energy cost}, and \figurename~\ref{figure: cost comparison under different data arrival - data cost}, respectively.

It can be observed from \figurename~\ref{figure: cost comparison under different data arrival - total cost} that the proposed algorithm outperforms all the baselines. The performance gain of the proposed algorithm over the PAS3 scheme demonstrates the necessity of periodic update of reference scheduling. The gains over the PAO and DA schemes demonstrate the benefit of dynamic programming within each super slot. The gain over the PPO scheme demonstrates the benefit of the proposed approximate MDP solution framework. 

It can be observed from \figurename~\ref{figure: cost comparison under different data arrival - time cost}, \figurename~\ref{figure: cost comparison under different data arrival - energy cost}, and \figurename~\ref{figure: cost comparison under different data arrival - data cost} that the proposed scheme achieves the best performance in the residual information bits, and good performance in the other two metrics. For example, the PAS3 scheme uses less time slots in transmission, but consumes more total energy, leading to larger average total cost.

Finally, it can be observed that the proposed algorithm can accomplish the delivery of computation tasks all the time, even when its size increases from $1000$Mb to $2000$Mb. The prices to pay is more transmission time slots and total energy consumption.

\begin{figure*}[!t]
    \centering
    \begin{subfigure}{0.23\linewidth}
        \centering
        \includegraphics[width=\linewidth]{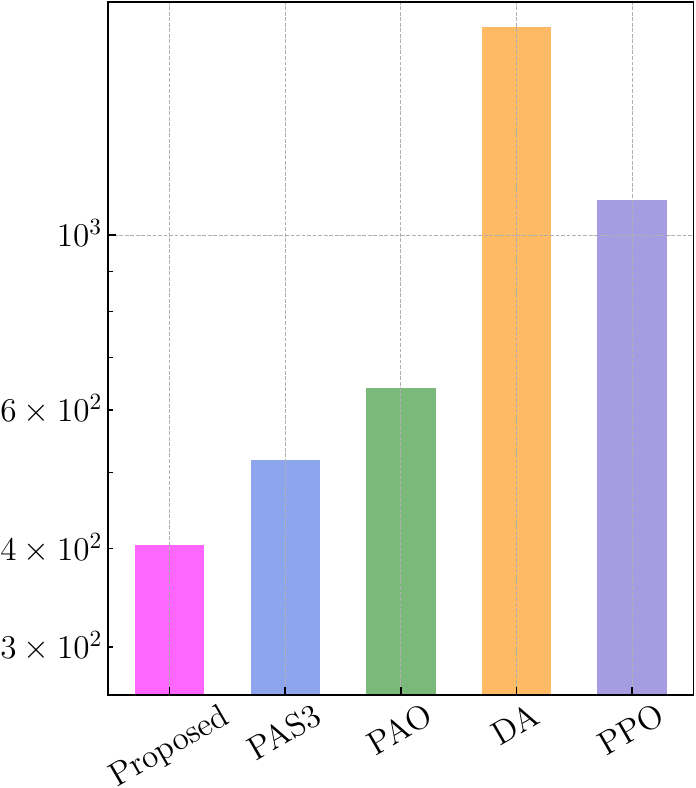}
        \caption{5 vehicles}
        \label{fig: average cost with increased vehicle number -- total cost}
    \end{subfigure}
    \begin{subfigure}{0.23\linewidth}
        \centering
        \includegraphics[width=\linewidth]{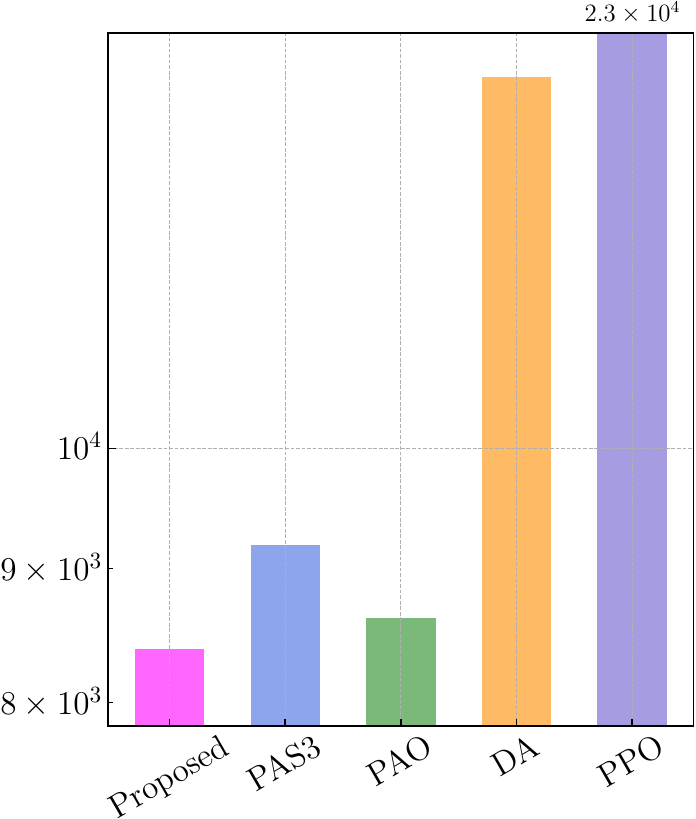}
        \caption{10 vehicles}
    \end{subfigure}
    \begin{subfigure}{0.236\linewidth}
        \centering
        \includegraphics[width=\linewidth]{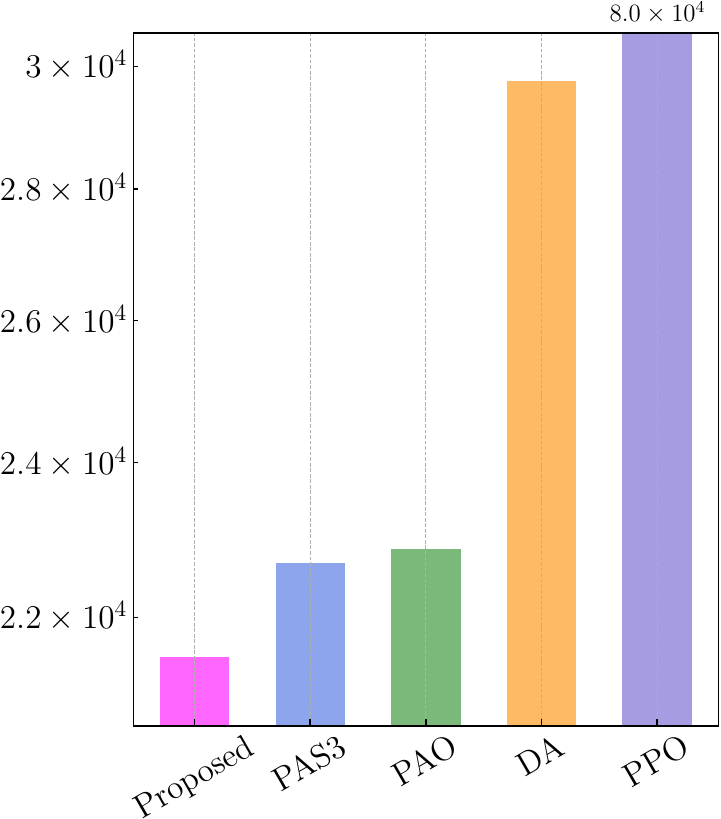}
        \caption{15 vehicles}
    \end{subfigure} 
    \vspace{1em}
    \begin{subfigure}{0.236\linewidth}
        \centering
        \includegraphics[width=\linewidth]{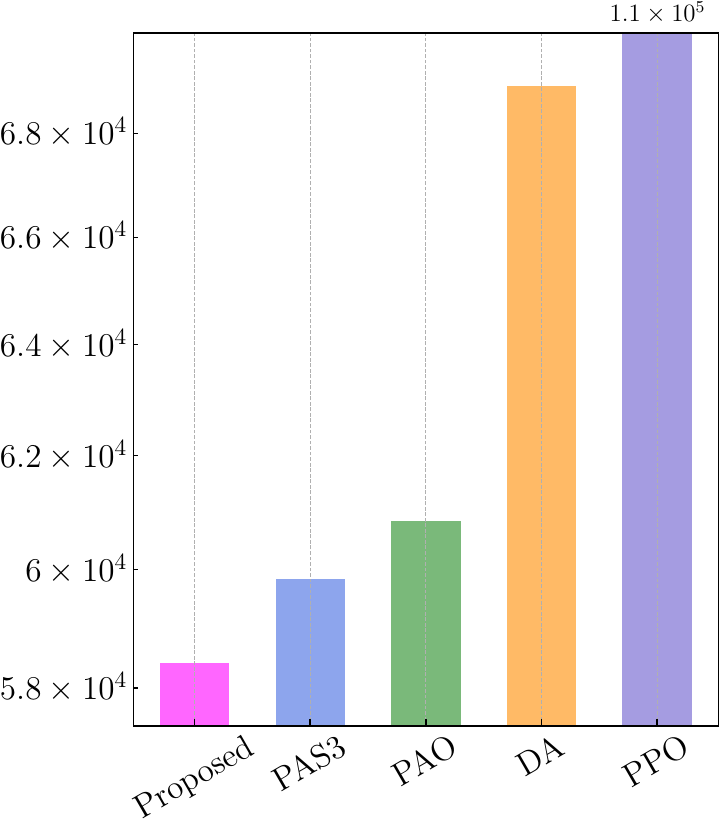}
        \caption{20 vehicles}
    \end{subfigure}
    \caption{The average total cost versus vehicle number.}
    \label{fig: average cost with increased vehicle number}
\end{figure*}
\begin{figure}[!t]
	\centering
	\includegraphics[width=0.3\textwidth]{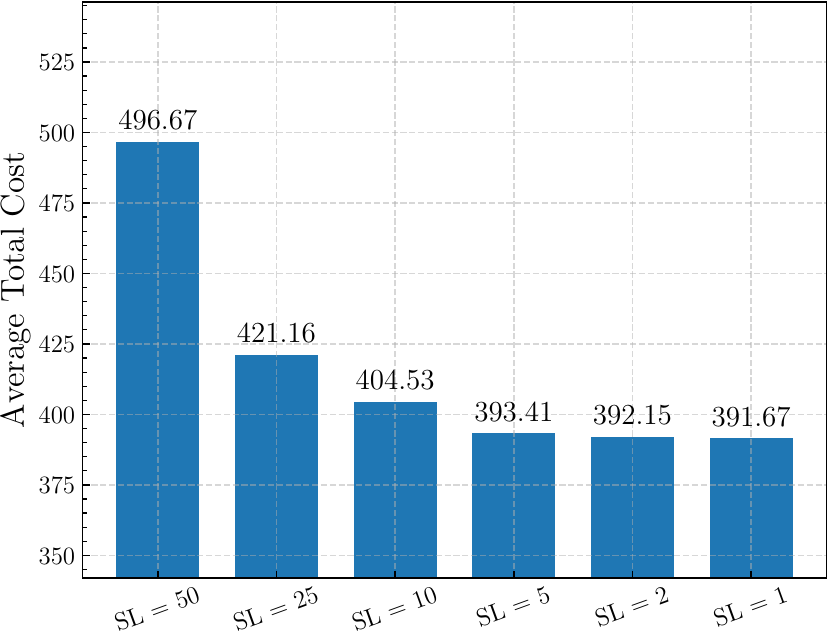}
	\caption{The average total cost versus super slot length.}
	\label{figure: comparison of changing super slot length}
\end{figure}

\subsection{Performance versus Task Arrival Time}
To evaluate how the task arrival timing affects the system performance, a scenario with five vehicles is simulated: the first vehicle receives its task at different time slots in different trials, while the other four vehicles receive their tasks at the first time slot.

Fig. \ref{figure: cost comparison under different arrival time slot} demonstrates the average total cost and its components (average energy consumption, average transmission time, and average residual information bits) versus the task arrival time of the 1-st vehicle. In the first 30 time slots, the impact on average total cost, average energy consumption, and average residual information bits remains trivial. 
However, the average transmission time decreases when the task of 1-st vehicle comes late. This is because the scheduler tends to complete the transmissions of the vehicles 2-5 first, and then allocate the remaining resources to the 1-st vehicle, resulting in shorter transmission time on average. In contrast, when the task arrives in the last 10 time slots, the average total cost increases significantly. This is because the 1-st vehicle cannot complete the task offloading in the scheduling period with a large probability, leading to high penalty.

More insights can be illustrated from Fig. \ref{fig: per slot power dynamic arrival}, which compares the scheduler behavior when the task of the 1-st vehicle arrives at the 30-th and 40-th time slots, respectively. The solid curves show the average per-slot transmission energy, while the shaded regions represent the 5-th to 95-th percentiles of instantaneous per-slot transmission energy of all simulation trials. In both cases, the proposed algorithm demonstrates lower power variation than the baselines.

\subsection{Performance versus Vehicle Number \& Super Slot Length }
In this part, the impacts of vehicle number and super slot length on the performance of the proposed algorithm are examined.
In Fig.~\ref{fig: average cost with increased vehicle number}, the average total costs with different vehicle numbers ($5, 10, 15, 20$) are shown. The performance gain of the proposed algorithm over the best-performing baseline increases with network size: from $114.3$ to $1420.2$ units of cost. This gain demonstrates that our proposed algorithm could better exploit the vehicles' mobility diversity. 

In \figurename~\ref{figure: comparison of changing super slot length}, the performance of the proposed algorithm versus different super slot lengths ($\Gamma=50, 25, ..., 1$) is demonstrated. It can be observed that shorter super slot leads to better performance, at the price of larger computation complexity. This is because the reference scheduling can better adapt to the random trajectories with shorter super slots. On the other hand, when the super slot length is below $10$, the performance gain becomes negligible. Hence, it is not necessary to update the reference scheduling frequently, as the online scheduling can adjust the actual transmission parameters according to the system state. This observation justifies our design of two-time-scale scheduling.
\subsection{Performance versus Weight}

\begin{figure*}[t]
    \centering
    \begin{minipage}{\textwidth}
    \centering
    \includegraphics[width=0.9\textwidth]{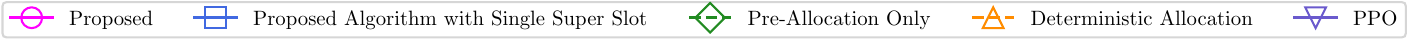}
    \end{minipage}
    
    \vspace{0.15cm}
    
    \begin{subfigure}{0.24\linewidth}
        \centering
        \includegraphics[width=\linewidth]{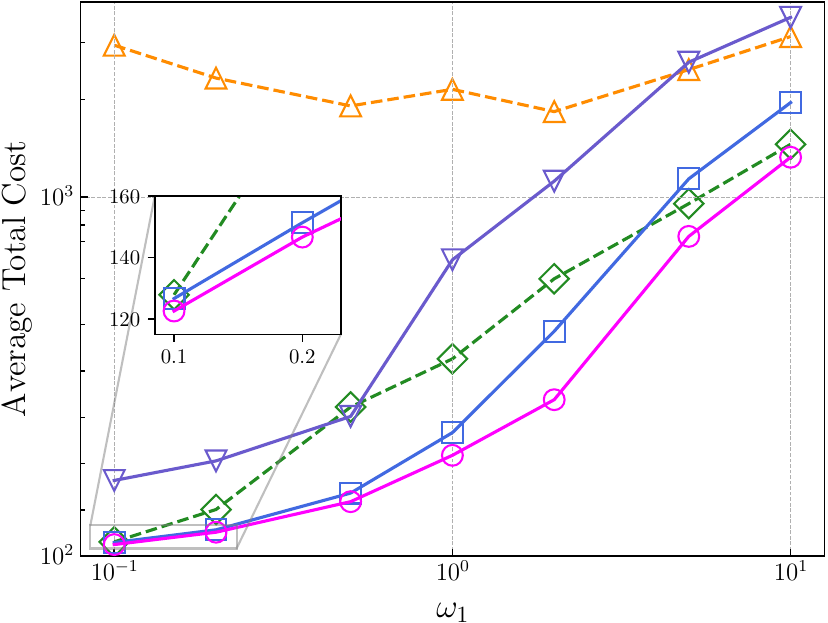}
        \caption{}
        \label{figure: cost comparison under different omega - total cost}
    \end{subfigure}
    \begin{subfigure}{0.24\linewidth}
        \centering
        \includegraphics[width=\linewidth]{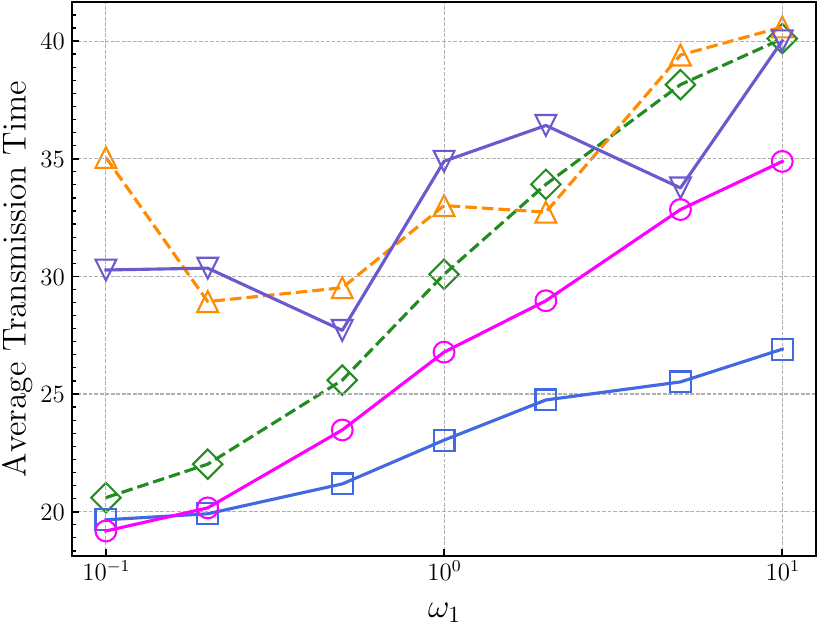}
        \caption{}
        \label{figure: cost comparison under different omega - time cost}
    \end{subfigure}
    \begin{subfigure}{0.24\linewidth}
        \centering
        \includegraphics[width=\linewidth]{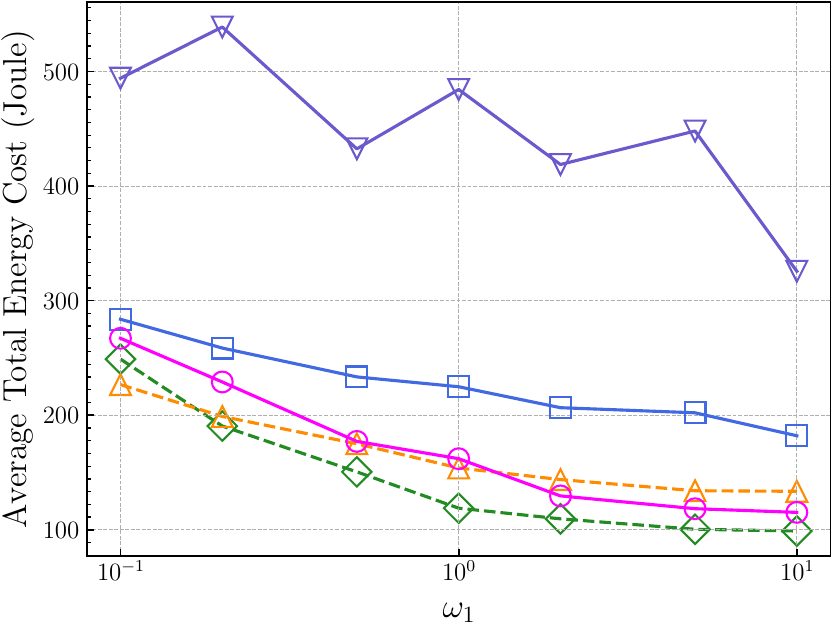}
        \caption{}
        \label{figure: cost comparison under different omega - energy cost}
    \end{subfigure}
    \begin{subfigure}{0.24\linewidth}
        \centering
        \includegraphics[width=\linewidth]{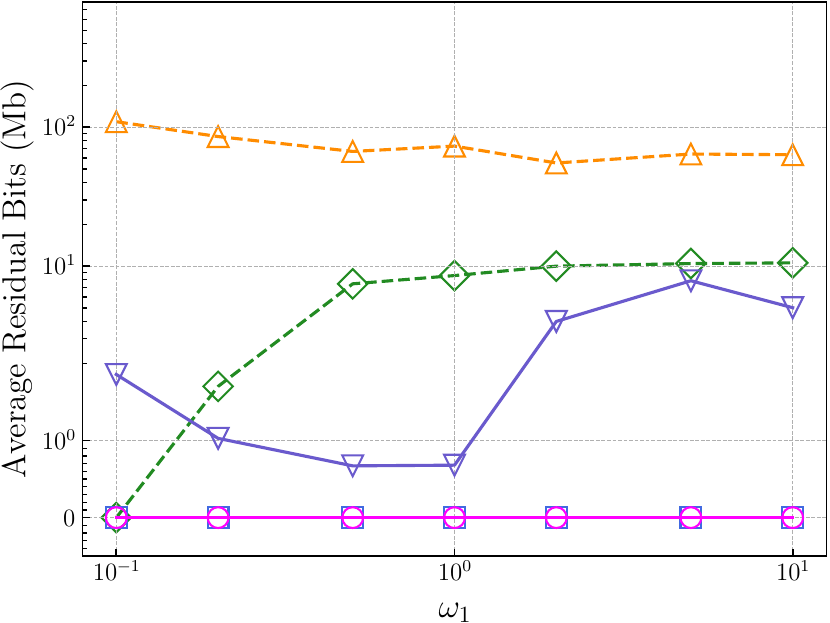}
        \caption{}
        \label{figure: cost comparison under different omega - data cost}
    \end{subfigure}
    \caption{The comparison of (a) average total cost, (b) average transmission time, (c) average energy cost and (d) average residual information bits per vehicle versus different weight values.}
    \label{figure: cost comparison under different omega}
\end{figure*}

The performance of the proposed algorithm and baselines with different weight values ($\omega_1 = 0.1, 0.2, 0.5, 1, 2, 5, 10$) is compared in \figurename~\ref{figure: cost comparison under different omega}, where $\omega_2 = 5$. When $\omega_1$ increases, the system puts larger weight on the energy consumption.

Hence, it can be observed that in the performance of the proposed algorithm and model-based baselines, larger $\omega_1$ leads to less total energy consumption, at the price of more transmission time slots. However, in Fig. \ref{figure: cost comparison under different omega - energy cost}, we can see that the energy consumption of the PPO scheme does not consistently decrease with increasing $\omega_1$. This is because PPO is a model-free reinforcement learning algorithm, whose performance depends on the randomness of training data. 

Moreover, it can be observed that the proposed algorithm outperforms all the baselines in most values of $\omega_1$, except that the performance of the proposed algorithm, PAO and PAS3 schemes are close with small value of $\omega_1$ ($\omega_1=0.1$). In this case, the scheduler is less sensitive to the energy consumption, and the benefit of two-time-scale scheduling is less significant. 

\subsection{Optimality Gap Analysis}
To assess the optimality gap of the proposed algorithm, the performance of the proposed algorithm and baselines is compared with that of the optimal policy.
\begin{figure}[!t]
    \centering
    \begin{subfigure}{0.49\linewidth}
        \centering
        \includegraphics[width=\linewidth]{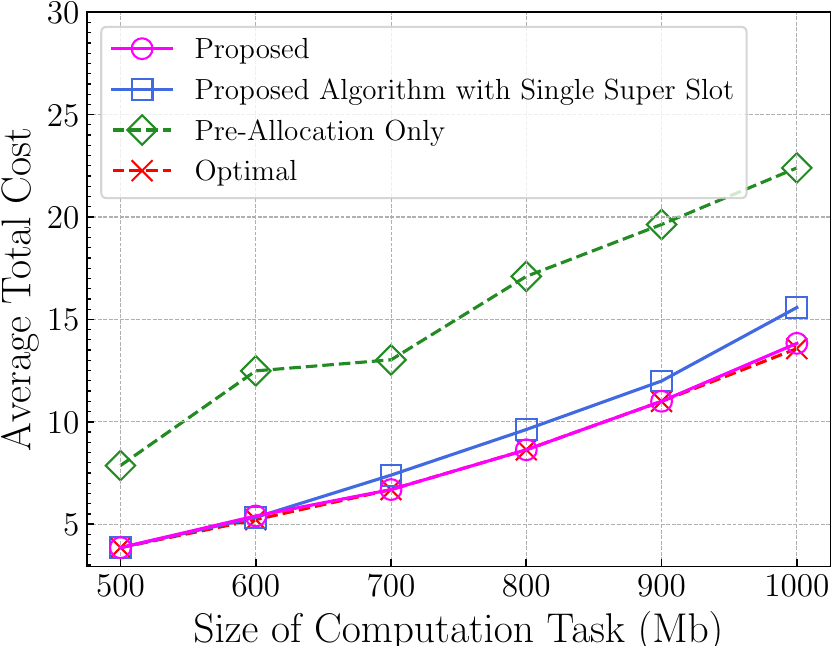}
        \caption{1 vehicle}
    \end{subfigure}
    \begin{subfigure}{0.49\linewidth}
        \centering
        \includegraphics[width=\linewidth]{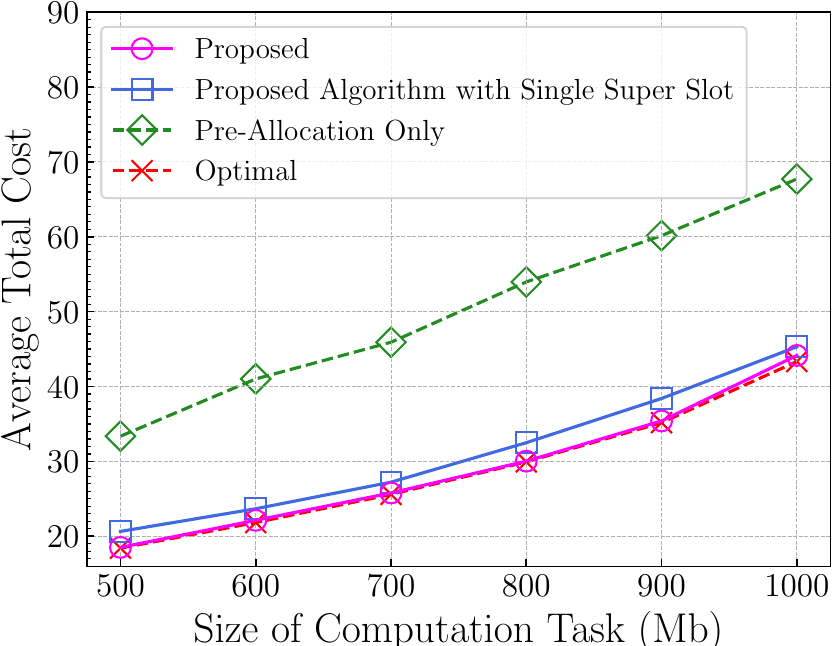}
        \caption{2 vehicles}
    \end{subfigure}
    \caption{Illustration of optimality gap with 1 and 2 vehicles.}
    \label{fig: optimality gap}
\end{figure}
In the optimal policy, we quantize the number of buffered information bits into $21$ levels. Similar quantization is also applied to the proposed algorithm and the baselines, in order to guarantee the fairness of comparison.
We evaluate the performance of scenarios with $1$ and $2$ vehicles under $T = 15$ and $\Gamma = 5$, where the task size ranges from $500$ Mb to $1000$ Mb. The corresponding state spaces reach $60,816$ and $427,308,714$ when the task size is $1000$ Mb, respectively. This highlights the curse of dimensionality, which justifies the necessity of the proposed low-complexity algorithm.
As shown in \figurename~\ref{fig: optimality gap}, the proposed algorithm maintains a small optimality gap. In the worst case, the average total cost of the proposed algorithm is $1.85\%$ higher than that of the optimal policy. 

Because of the prohibitive computational complexity, the optimal performance could not be simulated with more vehicles. Instead, we introduce an alternative lower-bound policy as follows:
\begin{itemize}
  \item \textbf{Perfect Prediction (PP)}: The scheduler optimizes the offloading decisions with perfect prediction of vehicles' trajectories.
\end{itemize}

\begin{figure}[!t]
    \centering
    \vspace{0.15em}
    \begin{subfigure}{0.51\linewidth}
        \centering
        \includegraphics[width=\linewidth]{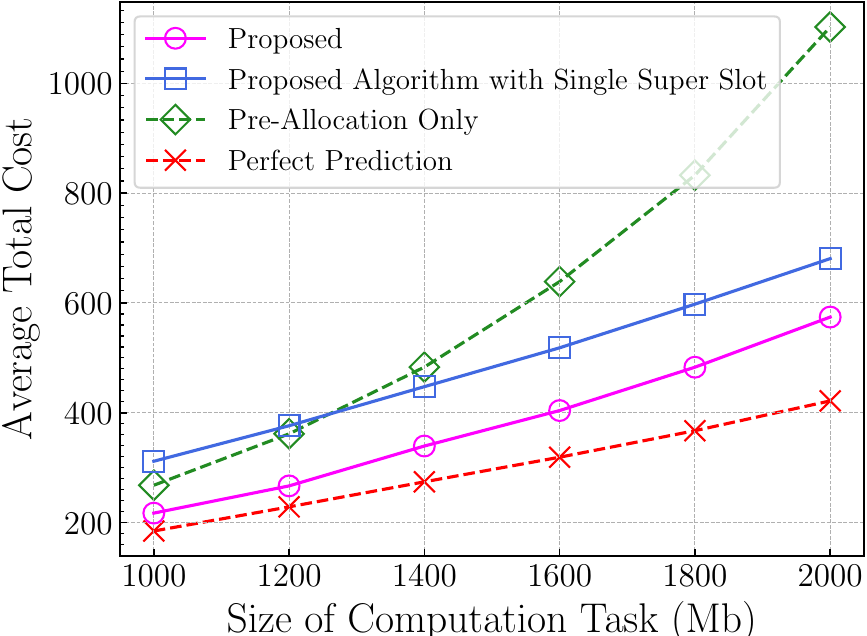}
    \end{subfigure}
    \vspace{0.42em}
    \caption{Illustration of optimality gap, where the average total cost versus different sizes of computation task is plotted.}
    \label{fig: optimality}
\end{figure}

Due to the perfect knowledge of the vehicles' trajectories, the PP scheme results in lower average total cost compared with the optimal policy. It can be observed from Fig.~\ref{fig: optimality} that the proposed algorithm consistently incurs higher average costs compared to the PP scheme. Prior knowledge on the future trajectories could suppress $15.02\%$-$26.59\%$ average cost compared with the proposed algorithm. The average total cost of the optimal scheduling policy should be larger than that of the PP scheme and smaller than that of the proposed algorithm.

\section{Conclusion}
\label{sec: Conclusion}
In this paper, the scheduling of task offloading in vehicular networks is investigated. The statistics of vehicles' mobility are learned from high-fidelity traffic simulator, and the uplink transmission of one task per vehicle is modeled as a finite-horizon MDP. To address the complexity issue of optimal solution, we propose a two-time-scale solution framework, where the optimal value functions are approximated in each super slot, and successive action improvement is conducted in each slot of online scheduling. An asymptotic upper bound of the average total cost is also derived. 
Simulation results  demonstrate that the proposed solution significantly outperforms the baseline algorithms.

\appendices

\section{Proof of The Lemma  \ref{lemma: Asymptotically Achievable Average Cost}}
\label{appendix: achievable cost}

In fact, the functions of $f^1_{n,t+1}$ and  $f^2_{n,t+1}$ represent the cost of transmission time plus residual information bits and the cost of energy consumption, respectively. Without loss of generality, take $k$-th feasible region as an example, where $$\sum_{\kappa= 1}^{k} r_{n,t+k} < d_{n,t+1} \le \sum_{\kappa= 1}^{k+1} r_{n,t+k}.$$ The average total cost since the $t$-th time slot can then be expressed as
\begin{align*}
    &\mathds{E}_{\state_{t+1},...,\state_{t+k+1} } \left[ \sum_{\kappa= 1}^{k+1} \mathsf{c}_{n, t+\kappa} (\state_{t+\kappa}, \action_{t+\kappa}) \middle| \state_t, \right] \\
    =& \mathds{E}_{ \state_{t+1},...,\state_{t+k+1}  } \left[ \sum_{\kappa= 1}^{k+1} \mathds{1}( d_{n,t+\kappa} > 0 ) + \omega_1 P_{n,t+\kappa} \txTime_{n,t+\kappa}  \middle| \state_t \right] \\
    =& k + \mathds{E}_{ \boldsymbol{l}_{n,t+1},...,\boldsymbol{l}_{n,t+k+1}  } \left[ \sum_{\kappa= 1}^{k+1} \omega_1 P_{n,t+\kappa} \txTime_{n,t+\kappa}  \middle| \boldsymbol{l}_{n,t} \right] \\
    =& k + \omega_1 \sum_{\kappa= 1}^{k+1} \mathds{E}_{\boldsymbol{l}_{n,t+\kappa}} \left[ P_{n,t+\kappa} \txTime_{n,t+\kappa} \middle| \boldsymbol{l}_{n,t} \right], 
\end{align*}
where the two terms refer to the costs of $f^1_{n,t+1}$ and  $f^2_{n,t+1}$ respectively.

To further compute $\mathds{E}_{\boldsymbol{l}_{n,t+\kappa}} \left[ P_{n,t+\kappa} \txTime_{n,t+\kappa} \middle| \boldsymbol{l}_{n,t} \right]$, we express the power $P_{n,t+\kappa}$, $\kappa \leq k$, as a function of the throughput $r_{n,t+k}$ as follows, where the high SNR approximation is applied.
\begin{equation*}
    P_{n,t+\kappa} \approx 2^{- \mathds{E}\left[\log_2 \frac{G\vert h \vert^2}{\sigma^2}\right]} \Vert \boldsymbol{l}_{n,t+\kappa} - \boldsymbol{l}_m \Vert_2^{\gamma} 2^{\frac{r_{n,t+\kappa}}{\txTime_{n, t}^{(\spIdx)} \slotLen \bandwidth}}.
\end{equation*}
For $\kappa = k + 1$, the vehicle transmits $d_{n, t} - \sum_{\kappa= 1}^{k} r_{n,t+k}$ information bits in $(t+\kappa)$-th time slot, with the transmit power 
\begin{equation*}
    P_{n,t+k+1} \approx 2^{- \mathds{E}\left[\log_2 \frac{G\vert h \vert^2}{\sigma^2}\right]} \Vert \boldsymbol{l}_{n,t+k+1} - \boldsymbol{l}_m \Vert_2^{\gamma} 2^{\frac{d_{n, t} - \sum_{\kappa= 1}^{k} r_{n,t+k}}{\txTime_{n, t}^{(\spIdx)} \slotLen \bandwidth}}.
\end{equation*}
Hence, the expression of $f^2_{n,t+1}$ can be obtained.

In order to facilitate the above transmission, the peak power constraint should satisfy 
$$P_{\max} \ge \overline{P}(\varUpsilon_{n,t+1}),$$
where
\begin{equation}
    \overline{P}(\varUpsilon_{n,t+1}) \triangleq \max_{ \substack{ r_n \in \varUpsilon_{n,t+1} \\  \boldsymbol{l_n} \in \mathcal{W}_n}} 2^{- \mathds{E}\left[\log_2 \frac{G\vert h \vert^2}{\sigma^2}\right]} \Vert \boldsymbol{l}_n - \boldsymbol{l}_m \Vert_2^{\gamma} 2^{\frac{r_n}{\txTime_{n, t}^{(\spIdx)} \slotLen \bandwidth}}.
\end{equation}
In the case of $\sum_{\kappa= 1}^{T - t} { r_{n,t+\kappa} } < d_{n,t+1}$, additional cost of residual information bits is included in $f^1_{n,t+1}$. This finishes the proof.
\qed

\section{Solution of \pref{Reference Policy - Throughput Optimization - Case 1} and \pref{Reference Policy - Throughput Optimization - Case 2}}
\label{Appendix: offline pre-allocation}

Without loss of generality, we provide the solution of \pref{Reference Policy - Throughput Optimization - Case 1} and \pref{Reference Policy - Throughput Optimization - Case 2} in the $i$-th iteration. 

\subsubsection{Solution of P7, $ \nthVehicleFinishTime \le T$ and $d_{n,T} = 0$} Let $$\gamma_{n,t} = \txTime_{n, t}^{i}\slotLen \bandwidth \log_2 \frac{P_{\max}}{\Phi_{n,t}}$$ denote the upperbound of scheduled uplink throughput $r_{n,t}$. 
The \textit{Karush-Kuhn-Tucker} (KKT) conditions \cite[Ch. 5]{boydConvexOptimization2004} for problem \pref{Reference Policy - Throughput Optimization - Case 1} are given by
\begin{equation}
	\begin{aligned}
		\omega_1' 2^{\frac{r_{n,t}}{\txTime_{n,t}^{i} \slotLen \bandwidth }} \Phi_{n,t} {\ln 2} - \lambda_{n,t} + \nu_{n,t} + \mu_n & = 0,    \\
		\lambda_{n,t} r_{n,t} &= 0, \\
        \nu_{n,t} \left( r_{n,t} -  \gamma_{n,t}\right) &= 0,  \\
		\sum_{t \in \mathcal{T}_n, t \le \nthVehicleFinishTime} r_{n,t} - D_{n,a} & = 0,   \\
		\lambda_{n,t}\ge 0, \nu_{n,t} &\ge 0.  \\
	\end{aligned}
\end{equation}
where $\omega_1' \triangleq \frac{\omega_1}{\slotLen \bandwidth}$, $\lambda_{n,t}$, $\nu_{n,t}$, $\mu_n$ are the Lagrange multipliers for the constraints of $r_{n,t} > 0 $, $r_{n,t} \le \gamma_{n,t} $ and $ \sum_{t \in \mathcal{T}_n, t \le \nthVehicleFinishTime} r_{n,t} = D_{n,a}$, respectively.
Therefore, we have
\begin{equation}
	\label{eqaution: reference rate compute 1}
	\mu_n = - \omega_1' \Phi_{n,t} \ln 2 \cdot 2^{\frac{r_{n,t}}{\txTime_{n,t} \slotLen \bandwidth}}, \forall r_{n,t} \in (0, \gamma_{n,t}).
\end{equation}

It can be observed that when
\begin{equation}
	\label{eqaution: reference rate compute 2}
	\mu_n \ge  - \omega_1' \Phi_{n,t} \ln 2,
\end{equation}
$\lambda_{n,t} \ge 0$ and the $n$-th vehicle will not transmit in the $t$-th time slot. When
\begin{equation}
	\mu_n \le  - \omega_1' 2^{\frac{\gamma_{n,t}}{\txTime_{n,t}^{i} \slotLen \bandwidth}} \Phi_{n,t} \ln 2,
\end{equation}
$ \nu_{n,t} \ge 0 $ and the $n$-th vehicle will transmit at the peak power in the $t$-th time slot.

Hence, the time slots can be classified into three sets as follows:
\begin{equation}
	\label{definition: time slot set 1}
	\mathcal{T}_{n}^1(\mu_n) \triangleq \left\{ t \middle|  \mu_n \ge  - \omega_1' \Phi_{n,t} \ln 2, \forall t \right\}, 
\end{equation}
\begin{equation}
	\label{definition: time slot set 2}
	\mathcal{T}_{n}^2(\mu_n) \triangleq \left\{ t \middle|  \mu_n \in \left(
		\begin{matrix}
			- \omega_1' 2^{\frac{\gamma_{n,t}}{\txTime_{n,t}^{i} \slotLen \bandwidth}} \Phi_{n,t} \ln 2, \\
			- \omega_1' \Phi_{n,t} \ln 2
		\end{matrix}\right), \forall t\right\},
\end{equation}
\begin{equation}
	\label{definition: time slot set 3}
	\mathcal{T}_{n}^3(\mu_n) \triangleq \left\{ t \middle|  \mu_n \le - \omega_1' 2^{\frac{\gamma_{n,t}}{\txTime_{n,t}^{i} \slotLen \bandwidth}} \Phi_{n,t} \ln 2, \forall t \right\} ,
\end{equation}

Therefore, given variable $\mu_n$, the optimal objective can be expressed as
\begin{equation}
	\label{equation: optimal objective for reference policy}
	v(\mu_n) \triangleq \sum_{t \le \nthVehicleFinishTime} v_t(\mu_n),
\end{equation}
and the vehicle-wise objective function is given by
\begin{equation}
	v_t(\mu_n) = \begin{cases}
		0, & t \in \mathcal{T}_{n}^1(\mu_n), \\
		- \frac{\mu_n \slotLen \bandwidth }{ \ln 2 }\txTime_{n,t}^{i}, & t \in \mathcal{T}_{n}^2(\mu_n), \\
		 2^{\frac{\gamma_{n,t}}{\txTime_{n,t}^{i} \slotLen \bandwidth }} \Phi_{n,t}\txTime_{n,t}^{i} , & t \in \mathcal{T}_{n}^3(\mu_n),
	\end{cases}
\end{equation}

According to the last constraint of P7,  $\mu_n$ can be obtained according to
\begin{equation}
	\label{eq: constraint on data transmission}
	D_{n,a} = \sum_{t \in \mathcal{T}_n^2 (\mu_n)} \txTime_{n,t}^{i} \slotLen \bandwidth \log_2 \frac{\mu_n}{ - \omega_1' \Phi_{n,t} \ln 2 } 
	+\sum_{t \in \mathcal{T}_n^3 (\mu_n)} \gamma_{n,t}.
\end{equation}
Since the right-hand side is monotonically decreasing with respect to $\mu_n$, $\mu_n$ can be found by the bisection search method. If multiple $\mu_n$ candidates exist, the maximum one is chosen due to the monotonically increasing equation \eqref{equation: optimal objective for reference policy}.

\subsubsection{Solution of P8, $\nthVehicleFinishTime = T$, and $0 < d_{n,T}$}
In this case, the optimization problem can be transformed into
\begin{equation}
	\begin{aligned}
		\min_{  r_{n,T_{n,a}}, \ldots, r_{n,T} } \quad & \omega_1 \sum_{t \in \mathcal{T}_n} 2^{\frac{r_{n,t}}{\txTime_{n,t}^i \slotLen \bandwidth}} \Phi_{n,t} \txTime_{n, t}^i  \\ 
		 & +  \omega_2 \left(D_{n,a} - \sum_{t \in \mathcal{T}_n} r_{n,t} \right) \\
		\text{s.t.}  \quad
		 & -r_{n,t} \le 0, \forall t,                                                                                                                                          \\
		 & \sum_{t \in \mathcal{T}_n} r_{n,t} < D_{n,a},  \\
		 & r_{n,t} \le \gamma_{n,t}, \forall t. \\
	\end{aligned}
\end{equation}
where the constant $T = \nthVehicleFinishTime$ in the original objective is omitted.
The KKT conditions for the above problem are given by
\begin{equation}
	\begin{aligned}
		\omega_1' 2^{\frac{r_{n,t}}{\txTime_{n,t}\slotLen \bandwidth}} \Phi_{n,t} \ln 2 - \lambda_{n,t} + \nu_{n,t} + \left(\mu_{n} - \omega_2\right) & = 0,  \\
		\lambda_{n,t} r_{n,t} = 0, \nu_{n,t} \left( r_{n,t} - \gamma_{n,t} \right) &= 0, \\
		\mu_n \left(\sum_{t \in \mathcal{T}_n} r_{n,t}  - D_{n,a}\right) & = 0, \\
		\lambda_{n,t} \ge 0, \nu \ge 0, \mu_n & = 0. \\
	\end{aligned}
\end{equation}
Therefore, the solution can be given by
\begin{equation}
	r_{n,t} = \begin{cases}
		0,  & \omega_1' \Phi_{n,t} \ln 2 \ge \omega_2, \\
		\gamma_{n,t}, & \omega_1' 2^{\frac{\gamma_{n,t}}{\txTime_{n,t}\slotLen \bandwidth}} \Phi_{n,t} \ln 2 \le \omega_2, \\
		\txTime_{n,t} \slotLen \bandwidth \log_2 \frac{\omega_2}{\omega_1' \Phi_{n,t} \ln 2}, & \text{otherwise}.
	\end{cases}
\end{equation}
\qed

\bibliographystyle{IEEEtran}
\bibliography{IEEEabrv,reference,reply1}
\end{document}